\documentclass[11pt]{article}

\usepackage{amsmath,amscd,amsfonts,eucal,latexsym,amssymb,mathrsfs}
\usepackage{amsxtra, amsthm, calc, graphics, color, bbm, accents, slashed}

\theoremstyle{definition}
\newtheorem{Definition}{Definition}[section]
\theoremstyle{plain}
\newtheorem{Theorem}[Definition]{Theorem}
\newtheorem{Proposition}[Definition]{Proposition}
\newtheorem{Lemma}[Definition]{Lemma}
\newtheorem{Corollary}[Definition]{Corollary}

\theoremstyle{remark}

\numberwithin{equation}{section}
\newcounter{remcount}

\def\cD{{\cal D}}

\def\cH{{\cal H}}

\def\cK{{\cal K}}

\def\cM{{\cal M}}

\def\cS{{\cal S}}

\def\cW{{\cal W}}

\def\bC{{\mathbb C}}

\def\bN{{\mathbb N}}

\def\bR{{\mathbb R}}

\def\l{\lambda}

\def\supp{{\text{supp}\,}}

\newcommand{\id}{\mathbbm{1}}
\renewcommand{\Im}{{\mathrm{Im}\,}}

\newcommand{\sgn}{{\mathrm{sgn}}}
\newcommand{\bdx}{{\boldsymbol{x}}}
\newcommand{\bdy}{{\boldsymbol{y}}}
\newcommand{\bdz}{{\boldsymbol{z}}}

\newcommand{\bdp}{{\boldsymbol{p}}}
\newcommand{\bdq}{{\boldsymbol{q}}}
\newcommand{\bds}{{\boldsymbol{\sigma}}}

\newcounter{propcount}

\newlength{\maxlabelwidth}

\newcommand{\inst}[1]{$^\textrm{#1}$ }

\begin{document}

\title{The fermionic massless modular Hamiltonian}
\author{Francesca La Piana\inst{1}, Gerardo Morsella\inst{2}}
\date{
\parbox[t]{0.9\textwidth}{\footnotesize{%
\begin{itemize}
\item[1] Department of Mathematics, University of Oslo, P.O.\ Box 1053, 0316 Blindern, Oslo (Norway), e-mail: franla@math.uio.no
\item[2] Dipartimento di Matematica, Universit\`a di Roma Tor Vergata, via della Ricerca Scientifica, 1 I-00133 Roma (Italy), e-mail: morsella@mat.uniroma2.it
\end{itemize}
}}
\\
\vspace{\baselineskip}
November 29, 2024}

\maketitle


\begin{abstract}
We provide an explicit expression for the modular hamiltonian of the von Neumann algebras associated to the unit double cone for the (fermionic) quantum field theories of the 2-component Weyl (helicity 1/2) field, and of the 4-component massless Dirac and Majorana fields. To this end, we represent the one particle spaces of these theories in terms of solutions of the corresponding wave equations, and obtain the action of the modular group on them. As an application, we compute the relative entropy between the vacuum of the massless Majorana field and one particle states associated to waves with Cauchy data localized in the spatial unit ball.
\end{abstract}

\section{Introduction}
A milestone of the theory of operator algebras, the Tomita-Takesaki modular theory associates a canonical one parameter group of automorphisms $\{ \sigma^\varphi_t\}_{t \in \bR}$, the \emph{modular group}, to every von Neumann algebra $\mathscr{M}$ equipped with a faithful normal state $\varphi$~\cite{Ta}. The \emph{modular Hamiltonian} is then the self-adjoint generator $\log \Delta_\varphi$ of the unitary group implementing the modular group in the GNS representation $\pi_\varphi$ of $\varphi$:
\[
 \Delta_\varphi^{it} \pi_\varphi(A) \Delta_\varphi^{-it} = \pi_\varphi(\sigma^\varphi_t(A)), \qquad A \in \mathscr{M},\,t \in \bR.
\]

In the algebraic approach to quantum field theory~\cite{Ha}, for every (open) region $O$ of Minkowski space time $\mathbb{R}^4$, we consider the von Neumann algebra $\mathscr{A}(O)$ generated by all the observables that are measurable in that region, or, more generally, the (larger) von Neumann algebra $\mathscr{F}(O)$ generated by all the (charged) fields localized in $O$. Under standard assumptions on the net $O \mapsto \mathscr{F}(O)$, (resp. $O \mapsto \mathscr{A}(O)$) the Reeh-Schlieder theorem states that, if both $O$ and its causal complement $O'$ have non empty interiors, the vacuum vector $\Omega$ is cyclic and separating for each algebra $\mathscr{F}(O)$ (resp. $\mathscr{A}(O)$), so the restriction to this algebra of the vacuum state $\omega= \langle \Omega, (\cdot)\Omega\rangle$ is normal and faithful, and we can consider the associated modular Hamiltonian $\log \Delta_O$.

 It is then natural to try to give an explicit description of the modular Hamiltonian for a certain region of Minkowski space in a given field theory. Apart from its intrinsic interest, which dates back to the very discovery of modular theory at the end of the 60s, this problem has attracted renewed attention due to the connection between the modular Hamiltonian and Araki's relative entropy~\cite{Ar}, and to the interest in the computation of information theoretic quantities in quantum field theory, which has been rapidly developing in the last decades (see, e.g., \cite{ABCH, CGP, CH, CF, CLR, Lo, Wi} and references therein). 
 
We have some geometric examples of determination of the modular Hamiltonian. Bisognano and Wichmann~\cite{BW1, BW2} described the modular operators for a wedge region in the representation of the vacuum in a model independent setting, starting from the observation that a certain Lorentz boost leaves the wedge invariant; the modular group is indeed induced by this boost. Then Buchholz~\cite{Bu} showed that the modular group of the future light cone, for massless free field theory in odd spatial dimensions $>1$, is induced by the group of global dilatations. Building on these results, and exploiting the geometric relationship between the unit double cone centered at the origin $O_1$ and the wedge, provided by the conformal symmetry of the theory, Hislop and Longo~\cite{HL} obtained the modular group associated with $O_1$ in the case of the massless free scalar field, in terms of a one parameter group of conformal symmetries preserving the double cone. These results have been generalised to the case of massless free fields with non-zero helicity in~\cite{Hi}.

Analogously to what has been done in the scalar field case in~\cite{LM}, in the present work we employ the results of~\cite{Hi} to compute the one particle modular Hamiltonian of the field algebra associated to the unit double cone  for the massless Weyl (i.e., helicity $\frac12$), Dirac and Majorana fields. For any of these theories, the modular Hamiltonian on the corresponding fermionic Fock space is then the second quantization of the one particle one. The action of the modular Hamiltonian on the (smooth, compactly supported) Cauchy data $\Psi_0 : \bR^3 \to \bC^4$ of the Dirac equation is given by
\[
    i \log \Delta_{O_1}\Psi_0(\bdx) = -\pi \big[(1-r^2) \partial_k - x_k\big] \gamma^0\gamma^k \Psi_0(\bdx),
\]
where $\gamma^0$ and $\gamma^k$ ($k=1,2,3$) are the Minkowski gamma matrices (see~\eqref{eq:gammamat}) and $r= \left| \bdx \right|$, $\bdx \in \bR^3$.

Our motivation for working with the Cauchy data of solutions of the Weyl or Dirac equations is that we hope that this setting can be of some use in the much more difficult open problem of computing the double cone modular Hamiltonian for the corresponding massive theories. This is suggested by the fact that, due to the canonical anticommutation relations, the C*-algebra (and possibly the von Neumann algebra)  generated by the Cauchy data of the Dirac field is mass independent. It is also interesting to remark that, conversely, in the scalar case detailed knowledge of $\log \Delta_{O_1}$ seems fundamental in order to obtain a direct proof of the mass independence of the algebra $\mathscr{A}(O_1)$~\cite{CM}. 

As an application of the above result, we compute, in the massless Majorana field case, the relative entropy between the restrictions to $O_1$ of the vacuum and of the one particle state $\omega_\Psi$ induced by a Majorana solution $\Psi$ of the Dirac equation with Cauchy data supported in the unit ball, obtaining
\[
S(\omega_\Psi \| \omega) =\frac1{8\pi^2} \int_{\bR^3} d\bdx\,(1-r^2) \langle \Psi, T_{00}(0,\bdx) \Psi\rangle,
\]
with
\[
T_{00}(x) =  \frac i 2: \psi^\dagger(x) \partial_0 \psi(x) - \partial_0 \psi^\dagger(x) \psi (x) : 
\]
the quantum energy density of the Majorana field $\psi$ (see the beginning of Sec.~\ref{sec:spacetime} for our notations on spinorial Wightman fields). 

 The above formula for the relative entropy is to be compared with the analogous one obtained in~\cite{LM}, where the canonical energy density of the massless Klein-Gordon field appears together with an extra term proportional to the field's square, and suggests that the traceless stress-energy tensor should appear there too (see also~\cite{Tu}). This is also coherent with the expression of $\log \Delta_{O_1}$ in terms of the generator of special conformal transformations and with the classical Noether theorem. A general discussion of relative entropy for CAR algebras has recently appeared in~\cite{GMV}.

The rest of the paper in organized as follows. In Sec.~\ref{sec:spacetime}, for the paper to be reasonably self-contained, we briefly review the main results of~\cite{Hi}, and we compute the action of the modular group of $O_1$ on certain one particle vectors of the form $\phi(f)\Omega$, $\phi$ the Weyl field. In Sec.~\ref{sec:cauchy}, we provide the description of the one particle Hilbert space of the Weyl field in terms of (Cauchy data of) solutions of the Weyl equation, called waves, using the correspondence between the test functions and the solutions obtained by convolution with the causal propagator; then we express the action of the modular group on the wave space, and compute the modular Hamiltonian. Finally, in Sec.~\ref{sec:Dirac},  we use the decomposition of the 4-component Dirac and Majorana fields in terms of 2-components Weyl fields to obtain the massless Dirac and Majorana fields' local modular Hamiltonians, which is then in turn used to compute the relative entropy $S(\omega_\Psi \| \omega)$. The Appendices~\ref{app:A} and \ref{app:B} contain some technical facts used in the main text.

Part of this work is the subject of the first author's Master thesis in Mathematics, University of Roma Tor Vergata (2022), done under the supervision of the second author.

\section{Weyl field modular group in the spacetime formulation}\label{sec:spacetime}
In this section we compute an explicit formula for the action of the modular group on the one particle vectors of the free helicity $1/2$ field in the spacetime formulation, which is not provided in~\cite{Hi} and which we will need later on to obtain the modular hamiltonian in the Cauchy data formulation.

We start by introducing some notation. We will only use lower indices to denote the components of 3-vectors $\bdx = (x_1,x_2,x_3) \in \bR^3$ and of 4-vectors $x = (x_0, \bdx) \in \bR^4$. Accordingly, we will set $\partial_\mu := \frac{\partial}{\partial x_\mu}$, $\mu=0,\dots, 3$. Unless otherwise stated, we employ the summation convention over repeated indices. As usual, $\bdp \cdot \bdx = p_j x_j$ and $px = p_0 x_0 - \bdp\cdot \bdx$ are the euclidean and Minkowski scalar product of the 3-vectors $\bdp, \bdx \in \bR^3$ and of the 4-vectors $p=(p_0,\bdp), x = (x_0, \bdx) \in \bR^4$ respectively. We indicate by $\cS(O, \bC^s)$ and $C^\infty_c(O,\bC^s)$, $s \in \bN$, respectively the spaces of Schwartz and smooth, compactly supported functions with values in $\bC^s$ and support contained in the open set $O \subset \bR^d$ ($d=3,4$). In particular, if $O$ is bounded, then $\cS(O, \bC^s)$ and $C^\infty_c(O,\bC^s)$ coincide. We use the following conventions for the Fourier transforms of Schwartz functions on $\bR^4$ and $\bR^3$ respectively:
\[
\hat f(p) := \int_{\bR^4} dx\,f(x) e^{i px}, \quad p \in \bR^4, \qquad \hat g(\bdp) := \int_{\bR^3} d\bdx \,g(\bdx) e^{-i \bdp \cdot \bdx}, \quad \bdp \in \bR^3.
\] 
The Fourier transform of vector valued functions is defined componentwise.

Elements $\Psi \in \bC^s$, called spinors, will be thought of as column matrices, and their indices will be denoted by greek letters. The standard norm on $\bC^s$ is denoted by $|\cdot|$. Matrix notation will be used whenever it does not cause confusion. In particular, $A^t$, $A^\dagger$ denote respectively the transpose and the hermitian conjugate of the matrix $A$. Similar notations will be used for spinorial Wightman fields (operator valued distributions): i.e., given a multiplet of Wightman fields $\psi_\alpha$, $\alpha = 1,\dots, s$, we will define $\psi(f) := \psi_1(f_1) + \dots+ \psi_s(f_s)$, $f \in \cS(\bR^4,\bC^s)$, and we will denote by $\psi(x)$ the column matrix of operator valued distributions with entries $\psi_\alpha(x)$, $\alpha = 1,\dots,s$, so that the formal equation
\[
\psi(f) = \int_{\bR^4} dx f(x)^t \psi(x)
\]
holds. As a consequence, $\psi^\dagger(x)$ will be the row matrix with entries $\psi_\alpha^*(x)$, $\alpha = 1,\dots,s$, the Wightman fields defined as usual by $\psi_\alpha^*(f) := \psi_\alpha(\bar f)^*$, $f \in \cS(\bR^4; \bC)$. So that, e.g., given an $s \times s$ complex matrix $A$, we will have
\[
\psi^\dagger(x) A \psi(y) = \psi^*_\alpha(x) A_{\alpha\beta}\psi_\beta(y)
\] 
as an operator valued distribution on $\bR^4\times \bR^4$.

As customary, we denote by $\sigma_0 = \id$ the $2 \times 2$ identity matrix, and by
\[
\sigma_1 = \left[\begin{matrix}0 &1 \\ 1 & 0\end{matrix}\right], \quad \sigma_2 = \left[\begin{matrix}0 &-i \\ i & 0\end{matrix}\right], \quad \sigma_3 = \left[\begin{matrix}1 &0 \\ 0 & -1\end{matrix}\right],
\]
the three Pauli matrices. To a 4-vector $x = (x_0, \bdx) \in \bR^4$ we associate the hermitian matrices
\[
\undertilde{x} := x_0 + \bdx \cdot \bds = x_0 \id + x_j \sigma_j, \qquad \tilde{x} := x_0 - \bdx \cdot \bds = x_0 \id - x_j \sigma_j ,
\]
for which $\det \undertilde x = \det \tilde x = x^2$ (Minkowski square). 

For the convenience of the reader, we now briefly summarize the main results of~\cite{Hi} (to which we refer for proofs and further details) which we will need in the following.

The free, helicity $1/2$, right-handed field on 4-dimensional Minkowski spacetime is given by a pair of (bounded) Wightman fields $\phi_\alpha$, $\alpha = 1,2$, acting on the fermionic Fock space $\Gamma_-(\cH)$ on the one particle Hilbert space $\cH = L^2(\bR^3)\oplus L^2(\bR^3)$. These fields satisfy in the distributional sense the right-handed Weyl equation
\[
\undertilde{\partial} \phi(x) = (\partial_0 + \sigma_j \partial_j)\phi(x)=0, 
\]
and canonical anticommutation relations (CAR)
\begin{equation}\label{eq:CAR}
\{ \phi_\alpha(x), \phi^*_\beta(y)\} = S_{\alpha\beta}(x-y), 
\end{equation}
where $\{\cdot,\cdot\}$ denotes the anticommutator, and $S := \tilde{\partial}D$, with $D$ the commutator (Pauli-Jordan) distribution of the Klein-Gordon field
\[
D(x) = \frac{i}{(2\pi)^3}\int_{\bR^4} dp\,\varepsilon(p_0)\delta(p^2)e^{-ipx}.
\]



The associated twisted local net of von Neumann algebras is defined by
\[
\mathscr{F}(O) := \{ \phi(f)+\phi(f)^*\,:\, f \in \cS(O, \bC^2)\}''
\]
with $O \subset \bR^4$ open, and $\phi(f) := \phi_1(f_1)+\phi_2(f_2)$. There is also on $\Gamma_-(\cH)$ a unitary, strongly continuous representation $U$ of the group $SU(2,2)$ of all complex $4 \times 4$ matrices $g$ such that
\[
g B g^\dagger = B, \quad \det g = 1, \qquad \text{with }B = \left[\begin{matrix} 0 &-i \\ i &0\end{matrix}\right] \text{ ($2 \times 2$ blocks)},
\]
which is a four-fold cover of the 4-dimensional conformal group.
The group $SU(2,2)$ acts on the tube $T := \bR^4+i V_+$ by $(g,z) \in SU(2,2) \times T \mapsto gz \in T$, where
\begin{equation}\label{eq:SU22action}
\undertilde{gz} = (a\undertilde z +b)(c \undertilde z +d)^{-1}, \quad \text{if }g = \left[\begin{matrix} a &b \\ c&d\end{matrix}\right].
\end{equation}
Moreover, the universal cover $SL(2,\bC) \ltimes \bR^4$ of the Poincar\'e group is identified with the subgroup of $SU(2,2)$ of the matrices
\[
p(a,y) := \left[\begin{matrix} a &\undertilde{y} (a^\dagger)^{-1} \\ 0 &(a^\dagger)^{-1} \end{matrix}\right], \quad (a,y) \in SL(2,\bC) \ltimes \bR^4,
\]
Then, the restriction to this subgroup of (the limit as $\Im z \to 0$ of) the action~\eqref{eq:SU22action} is the usual action of $SL(2,\bC)\ltimes \bR^4$ on Minkowski space, and the net $O \mapsto \mathscr{F}(O)$ is covariant with respect to the restriction of the representation $U$ to $SL(2,\bC)\ltimes \bR^4$. Furthermore, the representation of the translations $y \in \bR^4 \mapsto U(p(1,y))$ has spectrum contained in $\bar V_+$ (spectrum condition).

The action of $SU(2,2)$ is most easily described thanks to the existence of strongly analytic (vector valued) functions $z \in T \mapsto \phi_\alpha(z)\Omega \in \cH$, $z \in T \mapsto \phi^*_\alpha(z)\Omega \in \cH$, $\alpha = 1,2$, whose boundary value for $\Im z \to 0$ are the vector valued distributions $\phi_\alpha(x)\Omega$, $\phi_\alpha^*(x)\Omega$~\cite[Thm.\ 3.4]{Hi}, i.e., such that, for all $f \in \cS(\bR^4,\bC^2)$,
\begin{equation}\label{eq:limphiz}
\phi(f)\Omega = \lim_{\substack{y \to 0\\ y \in V_+}} \int_{\bR^4} dx f(x)^t \phi(x+iy)\Omega, \qquad \phi(f)^*\Omega = \lim_{\substack{y \to 0\\ y \in V_+}} \int_{\bR^4} dx f(x)^\dagger \phi^*(x+iy)\Omega.
\end{equation}
Indeed, the representation $U$ acts on these functions as~\cite[Thm.\ 3.10, Rem.\ 3.11]{Hi}
\begin{equation}\label{eq:SU22phi}
\begin{aligned}
U(g)\phi(z)\Omega &= [\det( \undertilde{z} c^\dagger+d^\dagger)]^{-2}(\undertilde{z} c^\dagger + d^\dagger) \phi(gz)\Omega, \\
U(g)\phi^*(z)\Omega &= [\det( c\undertilde{z} +d)]^{-2}(c\undertilde{z}  + d)^t \phi^*(gz)\Omega,
\end{aligned}\qquad z \in T, \,g \in SU(2,2).
\end{equation}
Moreover, thanks to the spectrum condition, the translation representation can be extended to the tube $T$, obtaining a strongly continuous semi-group of bounded operators $U(p(1,w))$, $w \in T$, such that $U(p(1,w))\phi(z)\Omega = \phi(z+w)\Omega$ for all $z,w \in T$.

The modular group of the von Neumann algebra associated to the unit double cone $O_1 = \{ x \in \bR^4\,:\,|x_0|+|\bdx| < 1\}$, with the vacuum $\Omega$ as standard vector, is then expressed in terms of the $SU(2,2)$ representation $U$ by  \cite[Thm.\ 4.8]{Hi}
\begin{equation}\label{eq:modular}
\Delta_{O_1}^{i\lambda} = U(e(2\pi\lambda)), \qquad e(\lambda) = \left[\begin{matrix} \cosh \frac\lambda2 & -\sinh \frac\lambda 2\\-\sinh \frac\lambda 2 &\cosh \frac\lambda2\end{matrix}\right]
\text{ ($2 \times 2$ blocks)}, \qquad \lambda \in \bR.
\end{equation}
The one parameter subgroup $e(\lambda)$ of $SU(2,2)$ acts on $z \in T$, through~\eqref{eq:SU22action}, as
\[
\undertilde{\nu_\lambda(z)} := \undertilde{e(\lambda) z} =  \tau(\lambda,z)^{-1}\left[z_0\cosh \lambda - \frac12 (1+z^2)\sinh \lambda + \bdz \cdot \bds\right],
\]
with
\begin{equation}\label{eq:gamma}
\tau(\lambda, z) = \frac12(1+\cosh \lambda)-\frac12 z^2(1-\cosh \lambda) -z_0 \sinh \lambda.
\end{equation}
Extending $\nu_\lambda$ to the boundary $\bR^4$ of $T$ one gets a (singular) map
\begin{equation}\label{eq:nul}
\nu_\lambda : \bR^4 \setminus S_\lambda \to \bR^4, \qquad x=(x_0,\bdx) \mapsto \nu_\lambda(x)=\frac1{\tau(\lambda,x)}\left(x_0 \cosh \lambda - \frac12 (1+x^2) \sinh \lambda, \bdx\right)
\end{equation}
with
\[
    S_{\lambda} = \{x \in \bR^4\,:\,\tau(\lambda,x)=\sinh^2(\lambda/2)\big[(x_0-\coth(\lambda/2))^2-|\bdx|^2\big]=0\}
\]
and it is known that $O_1 \subset \bR^4 \setminus S_\lambda$ is a $\nu_{\lambda}-$invariant set, so that the family ${\nu_{\lambda}}$, $\lambda \in \mathbb{R}$, acts on $O_1$ as a one-parameter group of smooth diffeomorphisms~\cite[Sec.\ 3]{HL}. Fig.~\ref{fig:modflow} illustrates the flow of $\nu_\lambda$, $\lambda \in \bR$ (which is covariant under spatial rotations), on the $x_0, x_1$ plane (note that $\nu_\l$ extends to a diffeomorphism of the conformal compactification of $\bR^4$~\cite{BGL}, but we will not need this in the following).
\begin{figure}[h!]
\setlength{\unitlength}{1cm}
\centering
\begin{picture}(12.7,12.7)(-6.35,-6.35)
\put(-6.35,-6.35){\includegraphics{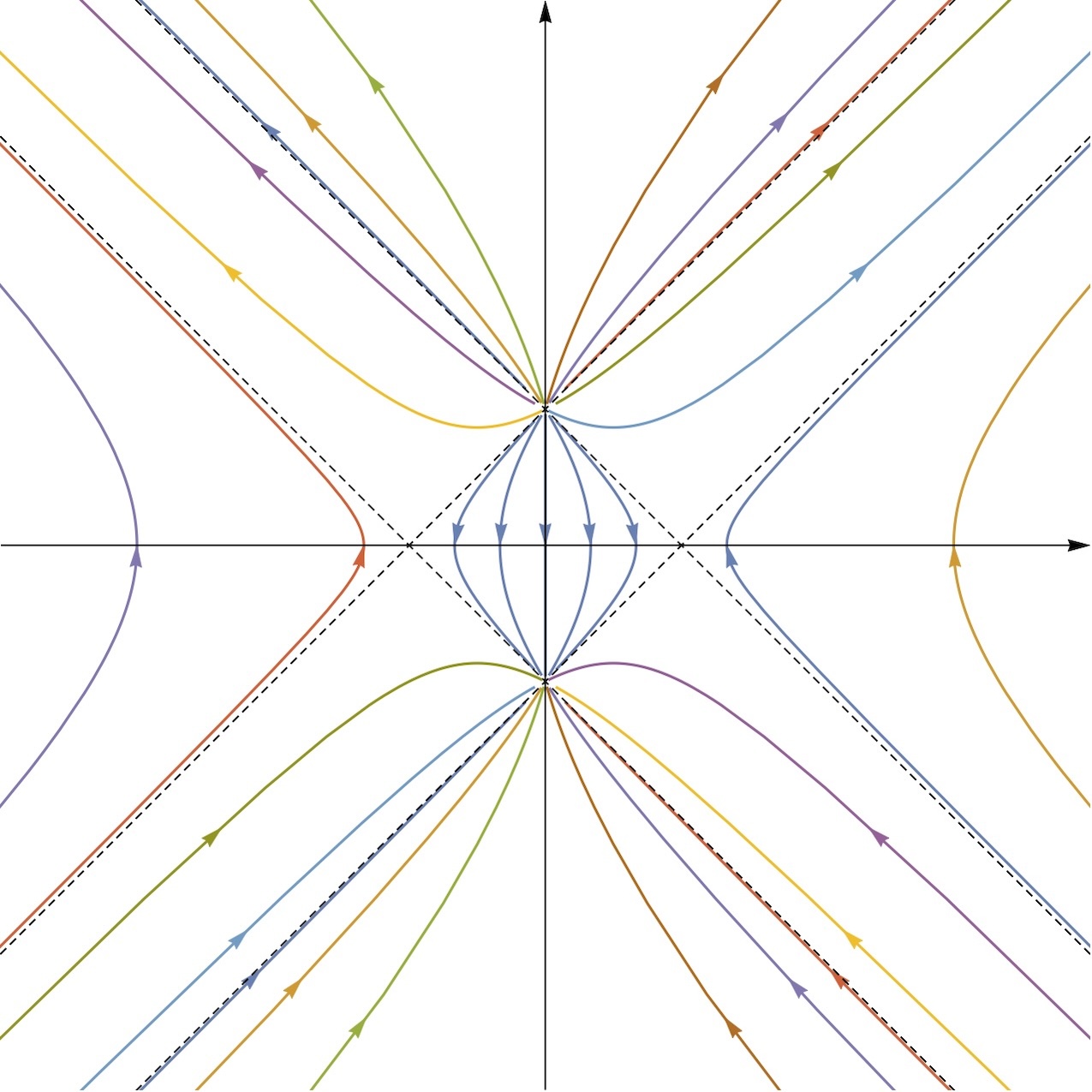}}
\put(1.5,-0.45){$1$}
\put(-1.85,-0.45){$-1$}
\put(6,0.2){$x_1$}
\put(0.2,6.1){$x_0$}
\end{picture}
\caption{\label{fig:modflow} The $O_1$ modular flow on the $x_0, x_1$ plane. 
The trajectories are continuous on the conformal compactification of Minkowski space. Those outside $O_1$ are singular on Minkowski, and their different branches share the same color.}
\end{figure}
\medskip

We now come to the announced formula for the action of the modular group. As a first step, we compute the Jacobian of the map $\nu_\l$.

\begin{Lemma}\label{lem:jacobian}
For any $\lambda \in \mathbb{R}$ and $x\in \bR^4\setminus S_\lambda$, the Jacobian determinant of $\nu_{\lambda}$ is given by
\begin{equation}\label{eq:jacobian}
    J_{\nu_{\lambda}}(x)=\frac{\sgn (\tau(\lambda,x))}{\tau(\lambda,x)^4} = \frac{\sgn (\tau(\lambda,x))}{\det(\cosh \frac\l2-\sinh \frac\l 2 \undertilde x)^4}.
\end{equation}
\end{Lemma}
\begin{proof}
It is convenient to introduce light cone spherical coordinates $(u,v,\theta,\varphi)$, where $\theta$, $\varphi$ are the spherical angles of $\bdx \in \bR^3$ and
\[
       u = x_0 + |\bdx|, \quad  v = x_0 - |\bdx|.
\]
In these coordinates, 
\[\begin{split}
        \tau(\lambda,x) & = \frac{1}{2}(1+\cosh \l)-\frac{1}{2}uv(1-\cosh \lambda)-\frac{u+v}{2}\sinh \l= \\ 
        &= \frac{1}{4e^\lambda}[(1+u)(1+v)+e^{2\lambda}(1-u)(1-v)+2e^{\lambda}(1-uv)]\\
        & = \frac{1}{4e^{\lambda}}[(1+u)+e^{\lambda}(1-u)]\cdot [(1+v)+e^{\lambda}(1-v)],
\end{split}\]
and the action of $\nu_\l$ is given by $\nu_\l(u,v,\theta,\varphi) = (u',v',\theta',\varphi')$, where for $\tau(\lambda, x) > 0$ one gets, from~\eqref{eq:nul}, $\theta'=\theta$, $\varphi'= \varphi$,
\[\begin{split}
u ' &= x_0'+|\bdx'| = \frac{1}{\tau(\lambda,x)}\left(x_0\cosh\lambda-\frac12(1+x^2)\sinh \lambda+|\bdx|\right) \\
&= \frac{1}{2\tau(\lambda,x)}[(u+v) \cosh \l-(1+uv)\sinh \l+u-v]\\
&= \frac{1}{4e^\lambda\tau(\lambda,x)}[(1+u)(1+v) - e^{2\lambda}(1-u)(1-v)+2e^\lambda (u-v)] = \frac{(1+u)-e^\lambda(1-u)}{(1+u)+e^\lambda(1-u)},
\end{split}\]
and, similarly,
\[
v' = \frac{(1+v)-e^\lambda(1-v)}{(1+v)+e^\lambda(1-v)},
\]
while for $\tau(\lambda, x) < 0$, being $|\bdx'| = - |\bdx|/\tau(\lambda,x)$, one gets
\[
u' = \frac{(1+v)-e^\lambda(1-v)}{(1+v)+e^\lambda(1-v)},\; v' = \frac{(1+u)-e^\lambda(1-u)}{(1+u)+e^\lambda(1-u)},\;\theta' = \pi - \theta, \;\varphi' = \pi+\varphi.
\]
By a straightforward computation the Jacobian determinant of the map $(u,v,\theta,\varphi) \mapsto (x_0,\bdx)$ is seen to be $\frac18(u-v)^2 \sin \theta$. Taking then into account that
\[
u'-v' = 2|\bdx'| = \frac{2|\bdx|}{|\tau(\lambda,x)|} = \frac{u-v}{|\tau(\lambda,x)|},
\]
one finds
\[\begin{split}
J_{\nu_\lambda}(x)&= \frac{(u'-v')^2}8 \sin \theta' \,\left(\frac{\partial u'}{\partial u} \frac{\partial v'}{\partial v}-\frac{\partial u'}{\partial v}\frac{\partial v'}{\partial u}\right) \frac{8}{(u-v)^2 \sin \theta}\\
& =\sgn( \tau(\lambda,x)) \frac1{\tau(\lambda,x)^2}\frac{4e^\lambda}{[(1+u)+e^\lambda(1-u)]^2}\frac{4e^\lambda}{[(1+v)+e^\lambda(1-v)]^2} = \frac{\sgn( \tau(\lambda,x))}{\tau(\lambda,x)^4}.
\end{split}\]
The last equality in the statement is then a simple computation.
\end{proof}

Observe that since $e(\lambda)$, $\lambda \in \bR$, is a one parameter group, one gets
\[
\undertilde{\nu_{-\lambda}(\nu_\lambda(z))} = \undertilde{e(-\lambda)e(\lambda)z} = \undertilde{z}, \qquad z \in T,
\] 
i.e., $\nu_{-\lambda}\cdot \nu_{\lambda}=\mathrm{id}$ on $T$. This entails that, for $x \in \bR^4\setminus S_\lambda$, the limit of $\nu_{-\lambda}(\nu_\lambda(x+iy)) = x+iy$ as $y \to 0$ inside $V_+$ belongs to $\bR^4$, and as a consequence 
$\nu_\l(\bR^4\setminus S_\lambda) \subset \bR^4\setminus S_{-\l}$ and $\nu_{-\lambda} \cdot \nu_\lambda = \mathrm{id}$ also on $\bR^4\setminus S_\lambda$. It follows that $J_{\nu_{-\lambda}}(\nu_\lambda(x)) J_{\nu_\lambda}(x) = 1$ for $x \in \bR^4\setminus S_\lambda$, so thanks to the previous lemma one deduces that $\tau(-\lambda, \nu_{\lambda}(x))$ and $\tau(\lambda,x)$ have the same sign and
\begin{equation}\label{eq:gammainv}
    \tau(-\lambda, \nu_{\lambda}(x))= \frac{1}{\tau(\lambda,x)} , \qquad x \in \bR^4\setminus S_\lambda, \, \lambda \in \bR.
\end{equation}

In view of the above, if $f \in C^\infty_c(\bR^4\setminus S_{2\pi\lambda},\bC^2)$, the function $f\circ \nu_{-2\pi\lambda}$, whose support is the compact set $\nu_{2\pi\lambda}(\supp f) \subset \bR^4\setminus S_{-2\pi\lambda}$, vanishes in a neighbourhood of $S_{-2\pi\lambda}$. Therefore for any given $\lambda \in \bR$, the formula
\begin{equation}\label{eq:Clambda}
        (E_\lambda f)(x)=\tau(-2\pi\lambda,x)^{-2}[c_\lambda -s_\lambda \undertilde{\nu_{-2\pi \lambda} (x)} ]^t f(\nu_{-2\pi \lambda} (x)), \qquad x \in \bR^4, 
\end{equation}
where $c_\lambda := \cosh(\pi \lambda)$, $s_\lambda := \sinh(\pi \lambda)$, provides a well defined linear map $E_\lambda : C^\infty_c(\bR^4\setminus S_{2\pi\lambda},\bC^2) \to C^\infty_c(\bR^4\setminus S_{-2\pi\lambda},\bC^2)$.

From now on, we will only consider the modular operator of the algebra $\mathscr{F}(O_1)$, which we will denote simply by $\Delta$.

\begin{Theorem}\label{thm:modularH}
For any $\lambda \in \bR$ and  $f \in C^\infty_c(\bR^4\setminus S_{2\pi\lambda},\bC^2)$, there holds
\begin{equation}\label{eq:modularphifO}
\Delta^{i\lambda} \phi(f)\Omega = \phi(E_\lambda f)\Omega, \qquad \Delta^{i\lambda} \phi(f)^*\Omega = \phi(E_\lambda f)^*\Omega,
\end{equation}
\end{Theorem}

\begin{proof} 
We consider the 
 vector-valued function of $y \in V_+$ defined by
\[
    I_y(f) :=  \int_{\mathbb{R}^4} dx f(x)^t \Delta^{i\lambda}\phi(x+iy)\Omega = \int_{\bR^4}dx\, f(x)^t \frac{c_\lambda-s_\lambda (\undertilde{x+iy})}{\det(c_\lambda-s_\lambda (\undertilde{x+iy}))^2} \phi(\nu_{2\pi\lambda}(x+iy))\Omega,
\]
where the last equality follows from~\eqref{eq:modular} and \eqref{eq:SU22phi}.
Then for fixed $y' \in V_+$ 
we have 
\[
    U(p(1,iy'))I_y(f)= 
    \int_{\mathbb{R}^4}dx\,f(x)^t \frac{c_\lambda-s_\lambda (\undertilde{x+iy})}{\det(c_\lambda-s_\lambda (\undertilde{x+iy}))^2}\phi(\nu_{2\pi \lambda}(x+iy)+iy')\Omega
\]
which, through the change of variables $x \mapsto \nu^{-1}_{2\pi\lambda}(x)$ becomes, thanks to Lemma~\ref{lem:jacobian},
\begin{multline*}
    U(p(1,iy'))I_y(f) = \\
    =\int_{\mathbb{R}
^4}dx \, f(\nu^{-1}_{2\pi \lambda} (x))^t \frac{c_\lambda-s_\lambda\undertilde{(\nu^{-1}_{2\pi \lambda} (x)+iy)}}{\tau(-2\pi\lambda,x)^{4} \det\left[c_\lambda-s_\lambda(\undertilde{\nu^{-1}_{2\pi \lambda} (x)+iy})\right]^2}     \phi(\nu_{2\pi \lambda}(\nu^{-1}_{2\pi \lambda} (x)+iy)+iy')\Omega.
\end{multline*}
Since the integrand function has compact support in $x$ independent of $y$, and is a continuous function of $x$, $y$, it can be bounded by a constant times the characteristic function of $\supp(f\circ\nu_{-2\pi\lambda})$ uniformly for $y$ in a neighbourhood of the origin. So, by the dominated convergence theorem, if $y \rightarrow 0$,
\begin{equation}
\label{eq:limy0}
    U(p(1,iy'))I_y(f) \xrightarrow[y\to 0]{}  \int_{\mathbb{R}^4} dx\, (E_\lambda f)(x)^t  \phi(x+iy')\Omega,
\end{equation}
where~\eqref{eq:jacobian} and~\eqref{eq:gammainv} were used.
On the other hand, 
\[
        U(p(1,iy'))I_y(f) =  U(p(1,iy'))\Delta^{i\lambda} \int_ {\mathbb{R}^4} dx f(x)^t\phi(x+iy) 
\]
that, by~\eqref{eq:limphiz} and the boundedness of $U(p(1,iy'))\Delta^{i\lambda}$,  converges to
$U(p(1,iy'))\Delta^{i\lambda}\phi(f)\Omega$ as $y \rightarrow 0$.
Hence, from this and \eqref{eq:limy0}
\[
    U(p(1,iy'))\Delta^{i\lambda}\phi(f)\Omega = \int_{\mathbb{R}^4} dx\, (E_\lambda f)(x)^t  \phi(x+iy')\Omega,
\]
that, thanks to the strong continuity of $U(p(1,iy'))$ in $y' \in V_+$, and the fact that $E_\lambda f$ is smooth and compactly supported, implies the first equation in~\eqref{eq:modularphifO} in the limit $y'\to 0, y' \in V_+$. The second equation is obtained by the same argument, taking into account that for $x\in \bR^4$ the matrix
\[
\tau(-2\pi\lambda, x)^{-2}[c_\lambda -s_\lambda\undertilde{\nu_{-2\pi\lambda}(x)}]
\]
is hermitian.
\end{proof}

\begin{Corollary}\label{cor:modularphif}
    For any $\lambda \in \bR$ and $f \in C^\infty_c(\bR^4\setminus S_{2\pi\lambda},\bC^2)$, there holds
\[
    \Delta^{i\lambda}\phi(f)\Delta^{-i \lambda}=\phi(E_\lambda f).
\]
\end{Corollary} 

\begin{proof}
From~\cite[eq.\ (2.22)]{Hi} we have
\begin{equation}\label{eq_72}
    \phi(f)=a(L_f)+a(H_f)^*
\end{equation}
for specific elements $H_f = (0,h_f), L_f= (l_f,0) \in \cH$ (see~\eqref{eq:hflf} below), and where $a(\cdot)$, $a(\cdot)^*$ are the usual CAR algebra annihilation and creation operators. Then, taking into account that the representation $U$ of $SU(2,2)$ is of second quantization type, we have
\[
    \Delta^{i\lambda}\phi(f)\Omega =[a(U(e(2\pi \lambda))L_f)+a(U(e(2\pi \lambda))H_f)^*]\Omega=a(U(e(2\pi \lambda))H_f)^*\Omega,
\]
and on the other hand
\[
    \phi(E_\lambda f)\Omega=a(H_{E_\lambda f})^*\Omega,
\]
so that, by the previous theorem,
    $U(e(2\pi \lambda))H_f=H_{E_\lambda f}$.
In the same way from the action of the modular group on $\phi(f)^*\Omega$ one obtains
    $U(e(2\pi \lambda))L_f=L_{E_\lambda f}$. As a consequence of these relations
\[
    \Delta^{i\lambda}\phi(f)\Delta^{-i\lambda}= a(L_{E_\lambda f})+a(H_{E_\lambda f})^*=\phi(E_\lambda f),
\]
and the proof is complete.
\end{proof}

\section{Weyl field modular group in the Cauchy data formulation}\label{sec:cauchy}
 We denote by $\cW$ the space of functions $\Phi \in C^\infty(\bR^4,\bC^2)$ which are solutions of the (right-handed) Weyl equation
\begin{equation}\label{eq:weylwave}
\undertilde{\partial} \Phi(x) = (\partial_0 + \sigma_j \partial_j)\Phi(x)=0,
\end{equation}
with Cauchy data $\Phi_0(\bdx) := \Phi(0,\bdx)$ such that $\Phi_0 \in C^{\infty}_c(\bR^3,\bC^2)$. The elements of $\cW$ will be called (helicity $1/2$, right-handed) waves, and we will frequently identify a wave with its Cauchy data. In particular, given a test function $f \in C^\infty_c(\bR^4, \bC^2)$, the function
\[
\Phi^f(x) := \int_{\bR^4} dy\, S(x-y) \overline{f(y)}, \qquad x \in \bR^4
\]
(convolution in the distribution sense), belongs to $\cW$.

We are going to identify $\cW$ with (a dense subspace of) the one particle space $\cH = L^2(\bR^3)\oplus L^2(\bR^3)$ of the helicity $1/2$ field $\phi$ as described in~\cite[Sec.\ 2]{Hi}, in such a way that $\Phi^f$ gets identified with the vector $(\phi(f)+\phi(f)^*) \Omega$.

To this end, we must first introduce some notation. For $\bdp \in \bR^3$ we denote by $p_\pm := (\pm |\bdp|,\bdp)$ the corresponding 4-vectors on the future and past light cones, and the matrices $\undertilde{p}_\pm ( = - \tilde{p}_\mp)$ satisfy
\begin{equation}\label{eq:ppm}
\undertilde{p}_\pm \undertilde{p}_\pm = \pm2|\bdp| \undertilde{p}_\pm, \qquad \undertilde{p}_+ \undertilde{p}_- = 0,\qquad \undertilde{p}_+ - \undertilde{p}_- = 2|\bdp|,
\end{equation}
showing that $\pm\undertilde{p}_\pm/(2|\bdp|)$ are complementary orthogonal projections on $\bC^2$, known as the projections on $\pm 1/2$ helicity spinors. Another identity which we will use in the following is
\begin{equation}\label{eq:sigma2p}
\sigma_2 \tilde p_\pm \sigma_2 = \overline{\undertilde{p}_\pm}.
\end{equation}
We will also need the spinor, defined for almost all $\bdp \in \bR^3$,
\begin{equation}\label{eq:nu0p}
\nu_0(\bdp) := 
 \frac1{2|\bdp|\cos(\theta/2)}
\left[\begin{matrix} 
|\bdp|+p_3  \\
p_1+ip_2 
\end{matrix}\right] = 
\left[\begin{matrix} \cos(\theta/2) \\ \sin(\theta/2) e^{i\varphi}\end{matrix}\right], 
\end{equation}
with $\theta \in [0,\pi), \varphi \in [0,2\pi)$ the polar angles of $\bdp$. One verifies that $\undertilde{p}_-\nu_0(\bdp) = 0$, i.e., $\nu_0(\bdp)$ is a helicity $1/2$ spinor. Since also $|\nu_0(\bdp) |=1$, one also obtains
\begin{equation}\label{eq:nu0dyad}
\nu_0(\bdp) \nu_0(\bdp)^\dagger = \left[ \nu_0(\bdp)_{\alpha 1}\overline{\nu_0(\bdp)_{\beta 1}}\right]_{\alpha, \beta =1,2 } = \frac1{2|\bdp|}\undertilde{p}_+.
\end{equation}

 We now denote by $\bar \cW$ the completion of the space of waves with respect to the $L^2$ norm of the Cauchy data
\begin{equation}\label{eq:norm}
\| \Phi \|^2 = \int_{\bR^3} d\bdx | \Phi_0(\bdx)|^2 =\frac1{(2\pi)^3} \int_{\bR^3} d\bdp | \hat \Phi_0(\bdp)|^2.
\end{equation}
Moreover, instead of considering the obvious complex structure on $\bar \cW$, we will consider it as a complex vector space where multiplication by the imaginary unit is defined by
\[
(\imath \Phi)\hat{}_0(\bdp) := \imath(\bdp) \hat \Phi_0(\bdp) , \qquad 
\imath(\bdp) := \frac{i}{|\bdp|}(\bdp \cdot \bds)=\frac{i}{|\bdp|} \left[\begin{matrix} 
 p_3 & p_1-ip_2 \\ 
p_1+ip_2 & -p_3 
\end{matrix}\right].
\]
As a consequence, thanks to the identities
\[
\undertilde{p}_+ \iota(\bdp) = i \undertilde{p}_+, \qquad \iota(\bdp)^\dagger \undertilde{p}_- = i \undertilde{p}_-,
\]
$\bar \cW$ becomes a complex Hilbert space with the scalar product
\[
\langle \Psi, \Phi \rangle = \frac 1{(2\pi)^3} \int_{\bR^3} \frac{d\bdp}{2|\bdp|}\left[ \hat\Psi_0(\bdp)^\dagger \undertilde{p}_+ \hat\Phi_0(\bdp) - \hat\Phi_0(\bdp)^\dagger
 \undertilde{p}_- \hat\Psi_0(\bdp) \right], \qquad \Psi, \Phi \in \bar \cW,
 \]
 which, in view of~\eqref{eq:ppm}, induces the norm~\eqref{eq:norm}.
 
 \begin{Proposition}\label{prop:TequivH}
The map $V : \bar\cW \to \cH = L^2(\bR^3)\oplus L^2(\bR^3)$, $\Phi \mapsto (l_\Phi, h_\Phi)$, with
\[
l_\Phi(\bdp) :=- \frac 1{(2\pi)^{3/2}} \nu_0(\bdp)^\dagger \hat{\Phi}_0(\bdp), \qquad h_\Phi(\bdp) := \frac1{(2\pi)^{3/2}} \hat{\Phi}_0(-\bdp)^\dagger \nu_0(\bdp), \qquad \bdp \in \bR^3,
\]
is unitary and satisfies $V \Phi^f = (\phi(f)+\phi(f)^*) \Omega$ for all $f \in C^\infty_c(\bR^4,\bC^2)$.
\end{Proposition}

\begin{proof}
To begin with, we notice that the entries of $\nu_0(\bdp)$ are bounded functions of $\bdp$ (see Eq.~\eqref{eq:nu0p}), so that $V$ maps $\bar \cW$ indeed into $\cH$. One then computes that, for $\Psi, \Phi \in \bar \cW$,
\[\begin{split}
\langle V\Psi, V\Phi \rangle &= \frac1{(2\pi)^3}\int_{\bR^3}d\bdp \left[\overline{l_\Psi(\bdp)}l_\Phi(\bdp) + \overline{h_\Psi}(\bdp)h_\Phi(\bdp)\right] \\
&= \frac1{(2\pi)^3} \int_{\bR^3}d\bdp \left[ \hat\Psi_0(\bdp)^\dagger \nu_0(\bdp) \nu_0(\bdp)^\dagger \hat\Phi_0(\bdp) + \overline{ \hat\Psi_0(-\bdp)^\dagger \nu_0(\bdp) \nu_0(\bdp)^\dagger \hat\Phi_0(-\bdp)}\,\right]\\
&= \frac1{(2\pi)^3} \int_{\bR^3}\frac{d\bdp}{2|\bdp|} \left[ \hat\Psi_0(\bdp)^\dagger \undertilde{p}_+ \hat\Phi_0(\bdp) + \overline{ \hat\Psi_0(-\bdp)^\dagger \undertilde{p}_+ \hat\Phi_0(-\bdp)}\,\right],
\end{split}\]
where in the last equality we used~\eqref{eq:nu0dyad}. Moreover, since $\undertilde{p}_+ \to -\undertilde{p}_-$ under the change of variables $\bdp \to -\bdp$, there holds
\[
 \int_{\bR^3}\frac{d\bdp}{2|\bdp|}  \overline{ \hat\Psi_0(-\bdp)^\dagger \undertilde{p}_+ \hat\Phi_0(-\bdp)} = -\int_{\bR^3}\frac{d\bdp}{2|\bdp|}  \hat\Phi_0(\bdp)^\dagger \undertilde{p}_- \hat\Psi_0(\bdp),
\]
which inserted into the previous equation shows that $\langle V\Psi,V\Phi\rangle = \langle \Psi, \Phi\rangle$, i.e., that $V$ is isometric. Moreover, $\psi \in \cH$ is in the range of $V$ if and only if $A(\bdp)\hat\Phi_0(\bdp) = [ \psi_1(\bdp) \; \overline{\psi_2(-\bdp)}]^t$ for some $\Phi \in \bar \cW$, where
\[
A(\bdp) = \left[\begin{matrix} -\nu_0(\bdp)^\dagger \\ \nu_0(-\bdp)^\dagger \end{matrix}\right] = 
\left[\begin{matrix}- \cos(\theta/2) &-\sin(\theta/2)e^{-i\varphi}\\ \sin(\theta/2) &-\cos(\theta/2)e^{-i\varphi} \end{matrix}\right].
\]
Then, since
\[
A(\bdp)^{-1} = \left[\begin{matrix} -\cos(\theta/2) &\sin(\theta/2)\\ -\sin(\theta/2) e^{i\varphi}&-\cos(\theta/2)e^{i\varphi} \end{matrix}\right]
\]
has bounded entries, we conclude that $V$ maps $\bar \cW$ onto $\cH$ and is then unitary.

 In order to prove the last statement, we first observe that for $f \in C^\infty_c(\bR^4, \bC^2)$,
 \[
(\phi(f)+\phi(f)^*) \Omega = (l_f, h_f) \in L^2(\bR^3) \oplus L^2(\bR^3), 
\]
with
\begin{equation}\label{eq:hflf}
l_f(\bdp) := -\frac 1{(2\pi)^{3/2}} \hat{\bar{f}}(p_+)_{\alpha}\overline{\nu_0(\bdp)_{\alpha }}, \qquad h_f(\bdp) := \frac 1{(2\pi)^{3/2}} \hat{f}(p_+)_{\alpha}\nu_0(\bdp)_{\alpha }.
\end{equation}
Then, we compute
\[
\hat \Phi^f_0(\bdp) = \frac1{2|\bdp|}\big[\undertilde{p}_+ \hat{\bar{f}}(p_+)-\undertilde{p}_- \hat{\bar{f}}(p_-)\big],
\]
which, in view of~\eqref{eq:ppm}, entails
\begin{equation}\label{eq:Psiff}
\undertilde{p}_+\big(\hat \Phi^f_0(\bdp) - \hat{\bar{f}}(p_+)\big) = 0 = \undertilde{p}_-\big(\hat \Phi^f_0(\bdp) - \hat{\bar{f}}(p_-)\big) .
\end{equation}
From the first one of the above equations, taking into account the hermiticity of $\undertilde{p}_+$ and the fact, observed above, that $\nu_0(\bdp)$ has helicity $1/2$, we then get
\[\begin{split}
\hat{\bar{f}}(p_+)_{\alpha}\overline{\nu_0(\bdp)_{\alpha }} &= \frac1{2|\bdp|}\hat{\bar{f}}(p_+)_{\alpha}\overline{(\undertilde{p}_+)_{\alpha \beta}}\overline{\nu_0(\bdp)_{\beta }} =  \frac1{2|\bdp|}(\undertilde{p}_+)_{\beta\alpha}\hat{\bar{f}}(p_+)_{\alpha}\overline{\nu_0(\bdp)_{\beta }} \\
&=\frac1{2|\bdp|}(\undertilde{p}_+)_{\beta\alpha}\hat{\Phi}^f_0(\bdp)_{\alpha}\overline{\nu_0(\bdp)_{\beta }} = \hat{\Phi}^f_0(\bdp)_{\alpha}\overline{\nu_0(\bdp)_{\alpha }},
\end{split}\]
which entails $l_f = l_{\Phi^f}$. Similarly, making the substitution $\bdp \to -\bdp$ in the second equation in~\eqref{eq:Psiff} one obtains $\hat{f}(p_+)_{\alpha}\nu_0(\bdp)_{\alpha } = \hat{\bar{\Phi}}^f_0(\bdp)_\alpha\nu_0(\bdp)_{\alpha }$ and therefore $h_f = h_{\Phi^f}$. 
\end{proof}

In view of the above discussion, we can now compute the action of the modular group on waves by pulling back its action on $\cH$ obtained in~Thm.~\ref{thm:modularH}. In the following, we will directly identify $\bar \cW$ with  $\cH$, avoiding to mention explicitly the unitary $V$, whenever this does not cause confusion.

\begin{Theorem}\label{thm:modularwave}
Given $\Phi \in \cW$ there is $\varepsilon > 0$ such that for $|\lambda| < \varepsilon$, $\Delta^{i\lambda} \Phi \in \cW$ and
\begin{equation}\label{eq:modularPsi}
\Delta^{i\lambda}\Phi(x) = \tau(-2\pi\lambda,x)^{-2} (c_\l +s_\l \undertilde x) \Phi(\nu_{-2\pi\lambda}(x)), \qquad x \in \bR^4 \setminus S_{-2\pi\lambda}.
\end{equation}
\end{Theorem}

\begin{proof}
For $\Phi \in \cW$ there is an open ball $B_r \subset \bR^3$ of radius $r > 0$ centered at the origin which contains the support of $\Phi_0$. Then, arguing as in the proof of~\cite[Lemma A.3]{Di}, we can find $f \in C^\infty_c(O_r,\bC^2)$, $O_r \subset \bR^4$ the double cone with basis $B_r$, such that $\Phi = \Phi^f$. Now for sufficiently small $|\lambda|$ one has $B_r \subset \bR^4 \setminus S_{2\pi\lambda}$, and as a consequence, by the above discussion and Thm.~\ref{thm:modularH}, $\Delta^{i\lambda} \Phi = \Phi^{E_\lambda  f}$, which in particular entails that $\Delta^{i\lambda} \Phi \in \cW$ is smooth on $\bR^4$. Moreover if $g \in C^\infty_c(\bR^4\setminus S_{-2\pi\lambda}, \bC^2)$, one has, thanks to~\eqref{eq:CAR}, Cor.~\ref{cor:modularphif} and~\eqref{eq:Clambda},
\[\begin{split}
    \int_{\bR^4} dx\:g(x)^t \Phi^{E_\lambda f}(x) &=  \langle \Omega, \{\phi(g), \phi(E_\lambda  f)^*\}\Omega \rangle =  \langle \Omega, \{\phi(g), \Delta^{i\lambda}\phi(f)^* \Delta^{-i\lambda}\}\Omega \rangle \\
    &=  \langle \Omega, \{\Delta^{-i\lambda}\phi(g)\Delta^{i\lambda}, \phi(f)^* \}\Omega \rangle = \langle \Omega, \{\phi(E_{-\lambda} g), \phi( f)^* \}\Omega \rangle \\
    &= \int_{\bR^4} dx\, (E_{-\lambda} g)(x)^t\Phi^f(x)\\
    &=\int_{\bR^4} dx\, \tau(2\pi\lambda,x)^{-2} g(\nu_{2 \pi \lambda}(x))^t [c_\lambda + s_\lambda\undertilde{\nu_{2\pi \lambda}(x)}]  \Phi^f(x).
\end{split}
\]
Using then the change of variable $x \to \nu_{-2\pi \lambda}(x)$, Lemma~\ref{lem:jacobian} and Eq.~\eqref{eq:gammainv}, we get
\[\begin{split}
    \int_{\bR^4} dx\:g(x)^t \Phi^{E_\lambda f}(x) &= \int_{\bR^4} dx \frac{|J_{\nu_{-2 \pi \lambda}}(x)| }{\tau(2\pi\lambda,\nu_{-2\pi\lambda}(x))^2} g(x)^t[c_\lambda+s_\lambda\undertilde{x}] \Phi^f(\nu_{-2 \pi \lambda}(x)) \\
    &=\int_{\bR^4} dx\: \tau(-2 \pi \lambda,x)^{-2}g(x)^t[c_\lambda+s_\lambda\undertilde{x}] \Phi^f(\nu_{-2 \pi \lambda}(x)),
\end{split}\]
and formula~\eqref{eq:modularPsi} follows by the arbitrariness of $g$.
\end{proof}

We can now compute the action of the modular hamiltonian  $ \log \Delta = -i \frac d{d\lambda} \Delta^{i\lambda} |_{\lambda = 0}$ on waves $\Phi \in \cW$. To begin with, we notice that from~\eqref{eq:gamma}, \eqref{eq:nul} it follows
\begin{align*}
\left.\frac d{d\lambda} \tau(-2\pi\lambda,x)\right|_{\lambda = 0} &= 2\pi x_0, \quad x \in \bR^4,\\
\left.\frac d{d\lambda} \nu_{-2\pi\lambda}(x)\right|_{\lambda = 0} &=\pi (1-x_0^2-|\bdx|^2,- 2x_0 \bdx), \quad x \in \bR^4\setminus S_{-2\pi\lambda}.
\end{align*}

\begin{Lemma} For any $x\in \bR^4$, 
we have:
\[
   \left. \frac{d}{d \lambda} \tau(-2\pi\lambda,x)^{-2} (c_\l +s_\l \undertilde x)\right|_{\lambda=0}=\pi(\undertilde{x}-4x_0).
\]
\end{Lemma}

\begin{proof}
By the above formulas, there holds
\[\begin{split}
 \left. \frac{d}{d \lambda} \tau(-2\pi\lambda,x)^{-2} (c_\l +s_\l \undertilde x)\right|_{\lambda=0}&=-2\tau(0,x)^{-3}\left.\frac d{d\lambda} \tau(-2\pi\lambda,x)\right|_{\lambda = 0} + \tau(0,x)^{-2} \left. \frac{d}{d \lambda} (c_\l +s_\l \undertilde x)\right|_{\lambda=0}\\
 &= -4\pi x_0 +\pi\undertilde{x}.
 \end{split}\]
\end{proof}

\begin{Theorem}\label{thm:modham} 
The one particle modular hamiltonian of the unit double cone for the helicity $1/2$ field on the wave space $\bar \cW$ is given by $\log \Delta = - \iota  K$, where for each $\Phi \in \cW$
\begin{equation}\label{eq:K}\begin{split}
    (K\Phi)_0(\bdx) &= -\pi \big[(1-r^2) \bds \cdot \nabla \Phi_0(\bdx) - \bdx \cdot \bds \Phi_0(\bdx) \big]
    \\
    &= -\pi\left[\begin{matrix}(1-r^2)\partial_3-x_3 & (1-r^2)(\partial_1-i\partial_2)-x_1+ix_2 \\
     (1-r^2)(\partial_1+i\partial_2)-x_1-ix_2 &-(1-r^2)\partial_3+x_3 \end{matrix}\right] \Phi_0(\bdx),
\end{split}\end{equation}
with $r=|\bdx|$.
\end{Theorem}

\begin{proof}
Given $\Phi \in \cW$, we need to check that 
\begin{equation}\label{eq:limK}
\lim_{\lambda \to 0}\left \| K\Phi-\frac{1}{\lambda}(\Delta^{i\lambda}\Phi -\Phi) \right \|^2 = \lim_{\lambda \to 0}\int_{\bR^3} d\bdx\, \left| (K\Phi)_0(\bdx)-\frac{1}{\lambda}\left[(\Delta^{i\lambda}\Phi)_0(\bdx) -\Phi_0(\bdx)\right] \right|^2 = 0.
\end{equation}
To this end, using the previous formulas and~\eqref{eq:modularPsi}, we compute, for fixed $x \in \bR^4 \setminus S_{-2\pi\lambda}$,
\[\begin{split}
\left. \frac d{d\lambda} \Delta^{i\lambda}\Phi(x) \right|_{\lambda = 0} &= \left. \frac{d}{d \lambda} \tau(-2\pi\lambda,x)^{-2} (c_\l +s_\l \undertilde x)\right|_{\lambda=0} \Phi(x) + \partial_\mu \Phi(x) \left.\frac d{d\lambda} \nu_{-2\pi\lambda}(x)_\mu\right|_{\lambda = 0} \\
&= \pi (\undertilde x-4x_0) \Phi(x) +\pi (1-x_0^2-|\bdx|^2) \partial_0 \Phi(x)-2\pi x_0 \bdx \cdot\nabla \Phi(x) \\
&= \pi (\undertilde x-4x_0) \Phi(x) -\pi (1-x_0^2-|\bdx|^2) \bds\cdot\nabla \Phi(x)-2\pi x_0 \bdx \cdot\nabla \Phi(x),
\end{split}\]
where in the last equality we used the Weyl equation~\eqref{eq:weylwave}. Thus for $|\bdx| \neq |\coth(\pi\lambda)|$,
\[
\left. \frac d{d\lambda} (\Delta^{i\lambda}\Phi)_0(\bdx) \right|_{\lambda = 0} = -\pi \big[(1-r^2) \bds \cdot \nabla \Phi_0(\bdx) - \bdx \cdot \bds \Phi_0(\bdx) \big],
\]
and as a consequence, to conclude the proof it is sufficient to show that the support of the integrand in~\eqref{eq:limK} is contained, for sufficiently small $|\lambda|$, in the ball of radius $|\coth(\pi\lambda)|$ centered at the origin, and that it is possible to interchange the limit and the integral. We show in Appendix~\ref{app:A} that this can actually be done by an application of the dominated convergence theorem.
\end{proof} 

\section{Modular group for massless Dirac and Majorana fields}\label{sec:Dirac}
If $\phi$ is a right-handed helicity $1/2$ field, it is immediate to verify that $\chi(x_0, \bdx) := \phi(x_0, -\bdx)$ is a left-handed helicity $1/2$ field, i.e., it satisfies the left-handed Weyl equation
\[
\tilde \partial \chi(x) = (\partial_0 - \sigma_j \partial_j)\chi(x) = 0,
\]
and the anticommutation relations
\[
\{ \chi_\alpha(x), \chi^*_\beta(y)\} = \undertilde{\partial}_{\alpha \beta} D(x-y).
\]
It is then immediate, using the results of the previous sections, to obtain the action of the modular group on $\chi$ and on the corresponding waves' space. We therefore avoid to write down explicit formulas.

If now $\phi$ and $\chi$ are independent right-handed and left-handed helicity $1/2$ fields respectively, with respective one particle spaces $\cH$ and $\cK$ (both copies of $L^2(\bR^3) \oplus L^2(\bR^3)$), it is easy to check that the 4-component field
\begin{equation}\label{eq:Dirac}
\psi(x) := \left[\begin{matrix} \phi(x) \\ \chi(x)\end{matrix}\right],
\end{equation}
acting on the fermionic Fock space over $\cH_D := \cH \oplus \cK$, satisfies the massless Dirac equation
\[
i \gamma^\mu \partial_\mu \psi(x) = 0
\]
with the (chiral) gamma matrices
\begin{equation}\label{eq:gammamat}
\gamma^0 = \left[\begin{matrix} 0 &-1 \\ -1 &0\end{matrix}\right], \qquad \gamma^k = \left[\begin{matrix} 0 &\sigma_k \\ -\sigma_k &0\end{matrix}\right], \quad k=1,2,3,
\end{equation}
and satisfies the anticommutation relations
\[
\{ \psi_\alpha(x),\psi^*_\beta(y)\}=(\gamma^\mu\gamma^0)_{\alpha\beta}\partial_\mu D(x-y), \qquad \alpha,\beta=1,\dots,4.
\]
The Dirac field then defines the twisted local net of von Neumann algebras
\[
\mathscr{F}_D(O) := \{ \psi(f)+\psi(f)^*\,:\, f \in \cS(O,\bC^4)\}'', \quad O \subset \bR^4 \text{ open}.
\]
It is then easy to obtain the description of the one particle space of the field $\psi$ in terms of classical solutions of the massless Dirac equation (Dirac waves), and the action on Dirac waves of the modular hamiltonian of the unit double cone, from the corresponding statements for the components Weyl fields $\phi$ and $\chi$, obtained above. We will therefore limit ourselves to state the main results, refraining from giving the straightforward proofs.

Let then $\cD$ denote the space of $C^\infty(\bR^4,\bC^4)$ solutions $\Psi$ of the massless Dirac equation with Cauchy data $\Psi_0 := \Psi(0,\cdot) \in C^\infty_c(\bR^3,\bC^4)$. Its completion $\bar \cD$ with respect to the Cauchy data $L^2$ norm is a complex Hilbert space with the scalar product
\[
\langle \Psi,\Phi\rangle= \frac1{(2\pi)^3}\int_{\bR^3} \frac{d\bdp}{2|\bdp|}\left[ \hat\Psi_0(\bdp)^\dagger \gamma^0 \slashed{p}_+\hat\Phi_0(\bdp)- \hat\Phi_0(\bdp)^\dagger \gamma^0 \slashed{p}_-\hat\Psi_0(\bdp)\right], \qquad \Psi, \Phi \in \bar \cD,
\]
where for $p \in \bR^4$, $\slashed{p}:=p_\mu \gamma^\mu$, and multiplication by the imaginary unit defined by
\[
(\imath\Psi)\hat{}_0(\bdp) := i\frac{p_k}{|\bdp|}\gamma^0\gamma^k \hat \Psi_0(\bdp).
\]
Using Prop.~\ref{prop:TequivH}, and its counterpart for the left handed Weyl fied, the space of Dirac waves $\bar \cD$ is then naturally identified with the one particle space $\cH_D$.

As a consequence of Thm.~\ref{thm:modham} (and of its left-handed version), the modular hamiltonian then acts on Dirac waves as follows.

\begin{Theorem}\label{thm:modhamDirac}
The one particle modular hamiltonian of the unit double cone for the massless Dirac field on the Dirac wave space $\bar \cD$ is given by $\log \Delta = - \imath  K_D$, where for each $\Psi \in \cD$
\begin{equation}\label{eq:KDirac}
    (K_D\Psi)_0(\bdx) = -\pi \big[(1-r^2) \partial_k - x_k\big] \gamma^0\gamma^k \Psi_0(\bdx).
\end{equation}
\end{Theorem}

It is also worthwhile to consider a massless Majorana field $\psi$, i.e., a massless Dirac field satisfying the reality condition
\[
\gamma^0 C \psi^*(x) = \psi(x), \qquad C := \left[\begin{matrix}- \sigma_2 &0 \\ 0 &\sigma_2\end{matrix}\right].
\]
As such, it can be obtained as the particular case of~\eqref{eq:Dirac} where $\cK = \cH$ and $\chi =  \sigma_2 \phi^*$. This implies the identity
\begin{equation}\label{eq:psiphi}
\psi(f)+\psi(f)^* = \phi(g + \sigma_2 \bar h) + \phi(g+\sigma_2 \bar h)^*, \qquad f = \left[\begin{matrix} g \\ h\end{matrix}\right], \; g, h \in C^\infty_c(\bR^4,\bC^2),
\end{equation}
showing that the local von Neumann algebras, and then the modular groups, of the massless Majorana and of the helicity 1/2 fields coincide; i.e., the massless Majorana field is just a 4-component reformulation of the Weyl field.

The corresponding one particle Hilbert space $\bar \cW$ can of course also be written in the 4-component formalism as
\[
\bar \cM := \{ \Psi \in \bar \cD\,:\, \Psi(x) = \gamma^0 C \overline{\Psi(x)} \;\forall x \in \bR^4\},
\]
which, thanks to the identity $C \gamma^0 \bar\gamma^k = \gamma^k \gamma^0 C$, is a complex closed subspace of $\bar \cD$ (i.e., $\Psi \in \bar \cM$ implies $\imath \Psi \in \bar \cM$), and it is also invariant under the operator $K_D$. Indeed, one easily verifies, using the identity~\eqref{eq:sigma2p}, that the unitary map identifying $\bar\cW$ with $\bar \cM$ is
\begin{equation}\label{eq:TequivSM}
\Phi \in \bar \cW \mapsto \frac1{\sqrt{2}} \left[ \begin{matrix} \Phi \\ \sigma_2 \bar \Phi\end{matrix}\right] \in \bar\cM,
\end{equation}
and it is immediate to see that it intertwines the action of $K$ in~\eqref{eq:K} with the restriction of that of $K_D$ in~\eqref{eq:KDirac} to $\bar\cM$. The latter provides then the modular hamiltonian of the unit double cone for the massless Majorana field.

As an application of our formula for the local modular hamiltonian of the Majorana field, the relative entropy of suitable one particle states with respect to the vacuum on the local algebra $\mathscr{F}(O_1)$ can be computed. To this end, if $\Psi \in \bar \cM$, we denote by $\omega_\Psi = \|\Psi\|^{-2}\langle \Psi, (\cdot)\Psi\rangle$ the corresponding vector state on the algebra of bounded operators on the antisymmetric Fock space on $\bar \cM$. In particular $\omega = \langle \Omega, (\cdot)\Omega\rangle$ is the vacuum state. We also denote by $\cM$ the set of smooth elements $\Psi \in \bar\cM$.

\begin{Proposition}\label{prop:relentropy}
Let $\Psi \in \cM$ be such that $\operatorname{supp} \Psi_0 \subset B$ and $\|\Psi\| = 1$. The relative entropy of $\omega_\Psi$ with respect to $\omega$ on $\mathscr{F}(O_1)$ is given by
\[
S(\omega_\Psi \| \omega) = \frac1{4\pi^2} \int \frac{d\bdp}{|\bdp|} \big[\hat\Psi_0(\bdp)^\dagger |\bdp|^2(1+\nabla^2) \hat\Psi_0(\bdp) + \hat\Psi_0(\bdp)^\dagger \bdp \cdot \nabla \hat\Psi_0(\bdp)+ i  p_j \hat\Psi_0(\bdp)^\dagger \sigma^{jk} \partial_k \hat\Psi_0(\bdp)\big],
\]
where $\sigma^{jk} = \frac i2[\gamma^j,\gamma^k]$.
\end{Proposition}

\begin{proof}
Analogously to the proof of Thm.~\ref{thm:modularwave}, from $\operatorname{supp} \Psi_0 \subset B$ it follows that we can find $f \in C_c^\infty(O_1,\bC^4)$ such that
\[
\Psi(x) = \Psi^f(x) := \int_{\bR^4} dy \,\gamma^\mu \gamma^0 \partial_\mu D(x-y) \overline{f(y)}.
\]
Moreover, one computes easily that $\gamma^0 C \overline{\Psi^f} = \Psi^{-\gamma^0 C \bar f}$, so that, since $\Psi = \Psi^f \in \bar \cM$, we can also choose $f$ such that $\gamma^0 C \bar f = -f$. This in turn entails
\[
f = \left[ \begin{matrix} g \\ -\sigma_2 \bar g\end{matrix}\right] \quad \text{and}\quad \Psi^f = \left[ \begin{matrix} \Phi^g \\ \sigma_2 \overline{\Phi^g}\end{matrix}\right] 
\]
with $g \in C_c^\infty(O_1,\bC^2)$, so that, thanks to~\eqref{eq:psiphi}, to Prop.~\ref{prop:TequivH} and to the identification~\eqref{eq:TequivSM},
\[
(\psi(f)+\psi(f)^*)\Omega = 2(\phi(g)+\phi(g)^*)\Omega \cong 2\Phi^g \cong \sqrt{2}\Psi^f.
\]
Moreover, from Eq.~\eqref{eq_72} and the CARs, one obtains
\[\begin{split}
\{\psi(f)+\psi(f)^*, \psi(f)+\psi(f^*)\} &= 4\{\phi(g)+\phi(g)^*,\phi(g)+\phi(g)^*\} = 8 \{a(l_g,h_g)^*,a(l_g,h_g)\}\\
&= 8\|(l_g ,h_g)\|^2 = 8 \|\Phi^g\|^2 = 4 \|\Psi^f\|^2 = 4,
\end{split}\]
i.e., the operator $B(f) :=2^{-1/2} (\psi(f)+\psi(f)^*) \in \mathscr{F}(O_1)$ is unitary. As a consequence, the unit vector $\Psi = B(f)\Omega$ is cyclic and separating for $\mathscr{F}(O_1)$, and for the one parameter group generated by the relative modular operator $\Delta_{\Psi, \Omega}$ one has~\cite[Lemma 5.7]{CGP}
\[
\Delta_{\Psi, \Omega}^{i\lambda} = B(f) \Delta^{i\lambda} B(f)^*, \qquad \lambda \in \bR.
\]
Therefore the relative entropy between $\omega_\Psi$ and $\omega$ on $\mathscr{F}(O_1)$ is given by
\[\begin{split}
S(\omega_\Psi \| \omega) &= -\langle \Omega, \log \Delta_{\Psi,\Omega}\Omega\rangle = i \frac d{d\lambda} \langle \Omega, B(f) \Delta^{i\lambda}B(f)^*\Omega\rangle\Big|_{\lambda = 0} \\
&= 2i \langle \Psi, K_D\Psi\rangle = -2 \Im \langle \Psi, K_D\Psi\rangle,
\end{split}\]
where the last equality is due to the skew-selfadjointness of $K_D$. The statement is then obtained by Lemma~\ref{lem:relentropy}.
\end{proof}

The formula for the relative entropy just obtained can be recast in a more familiar form by recalling that the quantum energy density of the Majorana field is expressed by
\begin{equation}\label{eq:T00}
T_{00}(x) = \frac i 2: \psi^\dagger(x) \partial_0 \psi(x) - \partial_0 \psi^\dagger(x) \psi (x) : = \frac i 2:\partial_k \psi^\dagger(x) \gamma^0 \gamma^k \psi(x) - \psi^\dagger(x) \gamma^0 \gamma^k \partial_k \psi(x) : 
\end{equation}
as a quadratic form on a domain which can be taken to be the linear span in $\Gamma_-(\bar\cM)$ of $\bigcup_{n=0}^{+\infty} \cM^{\otimes_a n}$ (antisymmetric tensor powers).
 
\begin{Theorem}
Let $\Psi \in \cM$ be such that $\operatorname{supp} \Psi_0 \subset B$ and $\|\Psi\| =1$. Then 
\[
S(\omega_\Psi \| \omega) =\frac1{4\pi^2} \int_{\bR^3} d\bdx\,(1-r^2) \langle \Psi, T_{00}(0,\bdx) \Psi\rangle.
\]
\end{Theorem}

\begin{proof}
Using the notation $t(\bdx) := \langle \Psi, T_{00}(0,\bdx) \Psi\rangle$ one gets, by Lemma~\ref{lem:T00},
\[\begin{split}
\hat t(\bdp) = \frac1{2(2\pi)^3} \int d\bdq &\bigg\{\big(|\bdq|+|\bdq+\bdp|\big) \hat \Psi_0(\bdq)^\dagger \hat \Psi_0(\bdq+\bdp)\\
&\quad+ \left( \frac 1{|\bdq|}+\frac1{|\bdq+\bdp|}\right) \hat\Psi_0(\bdq)^\dagger \big[\bdq\cdot(\bdq+\bdp)+i (\bdq \wedge \bdp)\cdot \boldsymbol{\Sigma}\big]\hat\Psi_0(\bdq+\bdp)\bigg\},
\end{split}\]
and therefore, by elementary computations and from Prop.~\ref{prop:relentropy},
\[
\int_{\bR^3} d\bdx\, (1-r^2) t(\bdx) =(2\pi)^3 (1+ \nabla^2) \hat t(\boldsymbol{0}) = 4\pi^2S(\omega_\Psi \| \omega),
\]
which proves the statement.
\end{proof}

The above result is coherent with the fact that $\log \Delta = \frac12( H-C_0)$, with $H$ the Hamiltonian and $C_0$ the time component of the generator of special conformal trasformations~\cite[Thm.\ 4.8]{Hi}, which is expected to be related to the energy density by
\[
C_0 = \int_{\bR^3} d\bdx \, r^2 T_{00}(0,\bdx)
\]
(see, e.g.,~\cite[Eq. (2.77)]{Gi}), by the classical Noether theorem and the tracelessness of the massless Majorana field energy-momentum tensor. However, to the best of our knowledge, no proof of this formula in quantum field theory is available in the literature.

We can also obtain the relative entropy between the restrictions of $\omega_\Psi$ and $\omega$ to the observable von Neumann algebra $\mathscr{A}(O_1)$ of the massless Majorana field. This is defined as the fixed point subalgebra of $\mathscr{F}(O_1)$ with respect to the involutive automorphism $\beta : \mathscr{F}(O_1) \to \mathscr{F}(O_1)$ such that
\[
\beta(\psi(f)) = -\psi(f), \qquad f \in C^\infty_c(\bR^4,\bC^4),
\]
i.e., $\beta = \operatorname{Ad} U(p(-1,0))$.

\begin{Corollary}
Let $\Psi \in \cM$ be such that $\operatorname{supp} \Psi_0 \subset B$ and $\|\Psi\| =\frac1{ \sqrt{2}}$. Then
\[
S\big(\omega_\Psi |_{\mathscr{A}(O_1)}\,\big\| \,\omega|_{\mathscr{A}(O_1)}\big) = S(\omega_\Psi \| \omega).
\]
\end{Corollary}

\begin{proof}
We denote by $E : \mathscr{F}(O_1) \to \mathscr{A}(O_1)$ the normal conditional expectation
\[
E(\psi) := \frac12(\psi + \beta(\psi)), \qquad \psi \in \mathscr{F}(O_1).
\]
Since $\omega \circ E = \omega$ and $\omega$ is a faithful state on $\mathscr{F}(O_1)$, $E$ is faithful too.  Moreover from
\[
U(p(-1,0))\Psi = U(p(-1,0))B(f)\Omega = - B(f)\Omega = -\Psi
\]
one gets $\omega_\Psi(\beta(\psi)) = \omega_\Psi(\psi)$ for all $\psi \in \mathscr{F}(O_1)$, i.e., $\omega_\Psi \circ E = \omega_\Psi$. The statement then follows from~\cite[Thm.\ 5.15]{OP}.
\end{proof}

An obvious extension of the above results would be to compute the relative entropy between the vacuum and a one particle state not localized in $B$. While it is fairly easy to verify that if such a state is localized in the complement of $B$ the relative entropy vanishes (e.g., since the relative modular operator coincides with $\Delta$), already for states which are sums of a state localized in $B$ and one in its complement the problem becomes much more difficult. We plan to come back to these issues in the future.

\section*{Acknowledgements}
We are pleased to thank R.\ Longo for his constant interest in this work, and for useful discussions. F.L.P. was partially supported by the Research Council of Norway (project 324944) and by the Trond Mohn Foundation (project 101510001). G.M. was partially supported by INdAM-GNAMPA, INdAM-GNAMPA project \emph{Operator algebras and infinite quatum systems} CUP E53C23001670001,  University of Rome Tor Vergata funding \emph{OAQM} CUP E83C22001800005, ERC Advanced Grant 669240 \emph{QUEST} and the MIUR Excellence Department Project \emph{MatMod@TOV} awarded to
the Department of Mathematics, University of Rome Tor Vergata, CUP E83C23000330006.

\section*{Declarations}
\textbf{Data availability.} Data sharing is not applicable to this article as no datasets were generated or analysed during the current study.
\smallskip

\noindent\textbf{Conflict of interest.} The authors have no competing interests to declare that are relevant to the content of this article. There are no financial and non-financial conflicts of interest.

\appendix

\section{Estimates for Thm.~\ref{thm:modham}}\label{app:A}
We show that it is possible to apply the dominated convergence theorem to interchange the limit and integral in~\eqref{eq:limK} in the proof of Thm.~\ref{thm:modham}. 

To this end, let $R > 0$ be such that $\supp \Phi_0 \subset B_R$.  It is then easy to see that there is an $\varepsilon > 0$ such that for $|\lambda| < \varepsilon$ there holds $\supp (\Delta^{i\lambda}\Phi)_0 \subset B_{2R}$. Indeed, in view of~\eqref{eq:modularPsi}, $\supp (\Delta^{i\lambda}\Phi)_0 = \{ \bdx \in \bR^3\,:\, \nu_{-2\pi\lambda}(0,\bdx) \in \supp \Phi\}$ and one has, thanks to Huyghens' principle, that $\supp\Phi$ is defined, in the light cone spherical coordinates introduced in the proof of Lemma~\ref{lem:jacobian}, by the inequalities $v \leq u$ and $-R \leq v \leq R$, or $-R \leq u \leq R$. On the other hand, still in these coordinates, one checks that on the time zero plane 
\[
\tau(-2\pi\lambda,(u,-u,\theta,\varphi)) > 0 \quad \Leftrightarrow \quad 0 \leq u < |\coth(\pi\lambda)|,
\]
so that $\nu_{-2\pi\lambda}$ acts on such plane as
\[
\nu_{-2\pi\lambda}(u,-u,\theta,\varphi) = \begin{cases}
(f_{-2\pi\lambda}(u), f_{-2\pi\lambda}(-u), \theta, \varphi), &\text{if }0 \leq u < |\coth(\pi\lambda)|,\\
(f_{-2\pi\lambda}(-u), f_{-2\pi\lambda}(u), \pi - \theta, \pi+ \varphi), &\text{if }u > |\coth(\pi\lambda)|,
\end{cases}
\]
with
\[
f_{\lambda}(u) := \frac{(1+u)-e^\lambda(1-u)}{(1+u)+e^\lambda(1-u)}.
\]
If then $\varepsilon$ is such that $\coth(\pi\varepsilon) > R$, and such that for $|\lambda| < \varepsilon$, there holds $ |\coth(\pi\lambda)| > \max\{ f_{2\pi\lambda}(R) , -f_{2\pi\lambda}(-R)\}$ and $f_{2\pi\lambda}(R) < 2R$ (which is possible since $f_{2\pi\lambda}(R) \to R$ as $\lambda \to 0$), one sees, by the monotonicity properties of $f_{-2\pi\lambda} = f_{2\pi\lambda}^{-1}$, that for $|\lambda| < \varepsilon$ 
\begin{align*}
&f_{2\pi\lambda}(R) < u < |\coth(\pi\lambda)| \quad \Rightarrow \quad f_{-2\pi\lambda}(u) > R,\, f_{-2\pi\lambda}(-u) < -R,\\
&u > |\coth(\pi\lambda)| \quad \Rightarrow \quad 
\begin{cases}
f_{-2\pi\lambda}(u) > f_{-2\pi\lambda}(\coth(\pi\lambda)) > R &\text{if }\lambda > 0,\\
f_{-2\pi\lambda}(-u) < f_{-2\pi\lambda}(\coth(\pi\lambda)) < -R &\text{if }\lambda < 0.
\end{cases}
\end{align*}
As a consequence $\supp (\Delta^{i\lambda}\Phi)_0$ is contained in the subset of the time zero plane defined by  $0 \leq u \leq f_{2\pi\lambda}(R)$, which is in turn contained in $B_{2R}$.

Since clearly $\supp (K\Phi)_0 \subset B_R$, it is then sufficient to show that for all $\bdx \in B_{2R}$ and for $|\lambda| < \varepsilon$,
\[ 
\left| (K\Phi)_0(\bdx)-\frac{1}{\lambda}[(\Delta^{i\lambda}\Phi)_0(\bdx) -\Phi_0(\bdx)] \right |\leq C,
\] 
where $C > 0$ is a constant independent from $\lambda$. In order to do this, we first observe that we can choose $\varepsilon > 0$ such that for $|\lambda| < \varepsilon$ there also holds $f_{4\pi\lambda}(R) < 2R$. This implies that if $|\lambda| < \varepsilon$, for all $\bdx \in \supp(\Delta^{i\lambda}\Phi)_0$ one has $\nu_{-2\pi\lambda}(0,\bdx) \in O_{2R}$, the double cone with basis $B_{2R}$. Indeed, if $0 \leq u \leq f_{2\pi\lambda}(R) < |\coth(\pi\lambda)|$, then
\[
f_{-2\pi\lambda}(u) \leq R < 2R, \quad f_{-2\pi\lambda}(-u) \geq f_{-2\pi\lambda}(-f_{2\pi\lambda}(R)) = -f_{4\pi\lambda}(R) > -2R,
\]
and in light cone spherical coordinates $O_{2R}$ is given by $- 2R < v \leq u < 2R$.

Then, from~\eqref{eq:K} and~\eqref{eq:modularPsi}, and using the last observation, we get the estimate
\[\begin{split}
    \bigg| (K\Phi&)_0(\bdx)-\frac{1}{\lambda}\left[(\Delta^{i\lambda}\Phi)_0(\bdx) -\Phi_0(\bdx)\right]\bigg| = \\
    &=\bigg|\pi \bdx \cdot \bds\big[\Phi_0(\bdx)-\Phi(\nu_{-2 \pi \lambda}(0,\bdx))\big]-\frac{1}{\l}\left[\Phi(\nu_{-2 \pi \lambda}(0,\bdx))-\Phi_0(\bdx)\right]+\pi(1-r^2)\partial_0 \Phi(0,\bdx)\\ 
    &\quad+\left[\pi \bdx \cdot \bds-\frac{1}{\lambda}\left(\frac{c_\l +s_\l \bdx \cdot \bds}{c_\l^2-r^2s_\l^2}-1\right)\right]\Phi(\nu_{-2\pi \lambda}(0,\bdx))\bigg| \\
    &\leq  \left[4\pi R \sum_{j=1}^3 \|\sigma_j\| +\bigg\|\pi \bdx \cdot \bds-\frac{1}{\lambda}\left(\frac{c_\l +s_\l \bdx \cdot \bds}{c_\l^2-r^2s_\l^2}-1\right)\bigg\|\right] \sup_{O_{2R}} |\Phi|  \\
    &\quad+ \bigg|\frac{1}{\l}\left[\Phi(\nu_{-2 \pi \lambda}(0,\bdx))-\Phi(0,\bdx)\right]-\pi(1-r^2)\partial_0 \Phi(0,\bdx)\bigg| ,
\end{split}\]
where as usual $r = |\bdx|$. Next, we can further estimate
    \[\begin{split}
         \bigg\|\pi &\bdx \cdot \bds-\frac{1}{\lambda}\left(\frac{c_\l +s_\l \bdx \cdot \bds}{c_\l^2-r^2s_\l^2}-1\right)\bigg\| \\
         &= \left\|\frac{1}{\lambda}(c_\l-1) +\frac{1}{\lambda}(s_\l  -\pi \lambda )\bdx\cdot\bds+\frac{1}{\lambda}(c_\l +s_\l \bdx \cdot \bds)\left(\frac{1}{c_\l^2-r^2s_\l^2}-1\right)\right\|\\
         &\leq \left|\frac{1}{\lambda}(c_\l-1)\right| +\left|\frac{1}{\lambda}(s_\l  -\pi \lambda )\right|\|\bdx\cdot\bds\|+(|c_\l |+|s_\l| \|\bdx \cdot \bds\|)\frac1{|c_\lambda^2-r^2 s^2_\lambda|}\left(\left|\frac1\lambda(c_\lambda^2-1)\right|+r^2 \frac{s_\lambda^2}{\lambda}\right),
    \end{split}\] 
and since  we can choose $\varepsilon$ such that if $|\lambda| < \varepsilon$ one has $\frac{2c_\lambda^2-1}{2s_\lambda^2} > 4R^2$, which entails $c_\l^2-r^2s_\l^2 \geq \frac{1}{2}$ for $r < 2R$, we see that the last line can be bounded by a constant uniformly for $|\lambda| < \varepsilon$ and $r < 2R$.

For the remaining term, by Lagrange's theorem we can find a $\xi_\lambda$ between $(0,\bdx)$ and $\nu_{-2\pi\lambda}(0,\bdx)$ such that, using also~\eqref{eq:nul},
\[\begin{split}
    \bigg|\frac{1}{\l}\big[&\Phi_\alpha(\nu_{-2 \pi \lambda}(0,\bdx))-\Phi_\alpha(0,\bdx)\big]-\pi(1-r^2)\partial_0 \Phi_\alpha(0,\bdx)\bigg|=\\
&= \bigg|(1-r^2)\partial_0\Phi_\alpha(\xi_\l)\left[\frac{s_\l}{\l}\frac{c_\l }{c_\l^2-r^2s_\l^2}-\pi\right]+x_k\partial_k\Phi_\alpha(\xi_\l)\frac1 \lambda \left(\frac{1}{c_\l^2-r^2s_\l^2}-1\right)\\
    &\quad+\pi(1-r^2)\left[\partial_0\Phi_\alpha(\xi_\l)-\partial_0 \Phi_\alpha(0,\bdx)\right]\bigg|\\
    &\leq (4R^2-1)\sup_{O_{2R}} |\partial_0\Phi_\alpha|\left(\left|\frac{s_\l}{\l}\frac{c_\l }{c_\l^2-r^2s_\l^2}-\pi\right|+2\pi\right)+2R \sup_{k,O_{2R}} |\partial_k\Phi_\alpha| \left|\frac1 \lambda \left(\frac{1}{c_\l^2-r^2s_\l^2}-1\right)\right|,
\end{split}\]
and, arguing as above, it is then easy to see that also this term can be estimated by a fixed constant uniformly for $|\lambda| < \varepsilon$ and $r < 2R$, thus concluding the proof.

\section{Computations for the relative entropy of the Majorana field}\label{app:B}

\begin{Lemma}\label{lem:relentropy}
For $\Psi \in \cM$, there holds
\[
\Im \langle \Psi, K_D\Psi\rangle =- \frac1{8\pi^2} \int \frac{d\bdp}{|\bdp|} \big[\hat\Psi_0(\bdp)^\dagger |\bdp|^2(1+\nabla^2) \hat\Psi_0(\bdp) + \hat\Psi_0(\bdp)^\dagger \bdp \cdot \nabla \hat\Psi_0(\bdp)+ i  p_j \hat\Psi_0(\bdp)^\dagger \sigma^{jk} \partial_k \hat\Psi_0(\bdp)\big].
\]
\end{Lemma}

\begin{proof}
From the identity $\slashed{p}_+ + \slashed{p}_- = 2 p_k \gamma^k$ one gets
\begin{equation}\label{eq:ImPKP}\begin{split}
\Im \langle \Psi, K_D\Psi\rangle &= \frac1{2i} \big( \langle \Psi, K_D\Psi\rangle - \langle K_D\Psi, \Psi\rangle\big) \\
&= \frac1{2i(2\pi)^3} \int_{\bR^3} \frac{d\bdp}{|\bdp|}p_k\big[ \hat\Psi_0(\bdp)^\dagger \gamma^0 \gamma^k \widehat{(K_D\Psi)}_0(\bdp) -   \widehat{(K_D\Psi)}_0(\bdp)^\dagger\gamma^0 \gamma^k \hat\Psi_0(\bdp)\big].
\end{split}\end{equation}
Moreover from Thm.~\ref{thm:modhamDirac} and using $\nabla^2( p_j \hat \Psi_0) = 2 \partial_j\hat\Psi_0 + p_j \nabla^2 \hat\Psi_0$ (where all the derivatives are with respect to the components of $\bdp$),
\[\begin{split}
\widehat{(K_D\Psi)}_0(\bdp) &= -\pi \gamma^0 \gamma^j \left[(1+\nabla^2)(ip_j \hat\Psi_0(\bdp))-i\partial_j \hat\Psi_0(\bdp)\right]\\
&= -i \pi \gamma^0 \gamma^j \left( p_j \nabla^2 \hat\Psi_0(\bdp) + \partial_j \hat\Psi_0(\bdp)+p_j \hat\Psi_0(\bdp)\right).
\end{split}\]
Inserting this formula in~\eqref{eq:ImPKP}, using the identities $(\gamma^0\gamma^j)^\dagger= \gamma^0 \gamma^j$, $\gamma^0 \gamma^j \gamma^0 \gamma^k = - \gamma^j \gamma^k = \delta_{jk} + i \sigma^{jk}$, and taking into account that $\sigma^{jk} = -\sigma^{kj}$, one then obtains
\[\begin{split}
\Im \langle \Psi, K_D\Psi\rangle =-\frac1{16\pi^2}\int_{\bR^3} \frac{d\bdp}{|\bdp|}\big[&|\bdp|^2\hat\Psi_0^\dagger(1+ \nabla^2) \hat\Psi_0 +\hat\Psi_0^\dagger\, \bdp\cdot \nabla \hat\Psi_0 +i \hat\Psi_0^\dagger p_k\sigma^{kj} \partial_j \hat \Psi_0\\
&+|\bdp|^2 (1+\nabla^2)\hat\Psi_0^\dagger\,\hat\Psi_0 + \bdp \cdot \nabla \hat\Psi_0^\dagger\, \hat\Psi_0 + i \partial_j\hat\Psi_0^\dagger \sigma^{jk} p_k \hat\Psi_0 \big].
\end{split}\]
Now, $\hat \Psi_0$ is a Schwarz function, and we can then integrate by parts the terms of the second line and use the elementary relations
\begin{gather*}
\nabla^2( |\bdp| \hat \Psi_0) = \frac 2{|\bdp|}\hat \Psi_0+\frac 2{|\bdp|}\bdp \cdot \nabla\hat\Psi_0+ |\bdp|\nabla^2\hat\Psi_0,\\
\nabla \cdot \left(\frac{\bdp}{|\bdp|}\hat \Psi_0\right) = \frac2{|\bdp|}\hat \Psi_0+\frac1{|\bdp|}\bdp \cdot \nabla \hat \Psi_0,\quad
\partial_j \left(\frac{p_k}{|\bdp|}\hat\Psi_0\right) = \frac{\delta_{jk}}{|\bdp|}\hat\Psi_0- \frac{p_k p_j}{|\bdp|^3}\hat\Psi_0 + \frac{p_k}{|\bdp|}\partial_j \hat\Psi_0,
\end{gather*}
to finally obtain the statement.
\end{proof}

\begin{Lemma}\label{lem:T00}
Given $\Psi \in \cM$,  with $T_{00}$ the quantum energy density of the Majorana field~\eqref{eq:T00} there holds
\[\begin{split}
\langle  \Psi, T_{00}(0,\bdx)\Psi\rangle = \frac1{2(2\pi)^6} \int_{\bR^6} d\bdp\, d\bdq e^{i(\bdq-\bdp)\cdot\bdx} &\bigg\{\big(|\bdp|+|\bdq|\big) \hat \Psi_0(\bdp)^\dagger \hat \Psi_0(\bdq)\\
&\;+ \left( \frac 1{|\bdp|}+\frac1{|\bdq|}\right) \hat\Psi_0(\bdp)^\dagger \big[\bdp\cdot\bdq+i (\bdp \wedge \bdq)\cdot \boldsymbol{\Sigma}\big]\hat\Psi_0(\bdq)\bigg\},
\end{split}\]
where $\Sigma_h := \frac12 \varepsilon_{hjk}\sigma^{jk}$. In particular $\langle \Psi, T_{00}(0, \cdot)\Psi\rangle$ is a smooth compactly supported function, whose support is contained in the one of $\Psi_0$.
\end{Lemma}

\begin{proof}
By the argument in the proof of Prop.~\ref{prop:relentropy}, we can write $\Psi = \sqrt{2}a(l_{\frac12}, l_{-\frac 12})^*\Omega$, with
\[
l_{\frac12}(\bdp) := l_\Phi(\bdp) = - \frac 1{(2\pi)^{3/2}} \nu_0(\bdp)^\dagger \hat{\Phi}_0(\bdp), \; l_{-\frac12}(\bdp) := h_\Phi(\bdp) = \frac1{(2\pi)^{3/2}} \hat{\Phi}_0(-\bdp)^\dagger \nu_0(\bdp), \; \Psi = \left[ \begin{matrix} \Phi \\ \sigma_2 \overline{\Phi}\end{matrix}\right] .
\]
On the other hand, from~\cite[Eq.\ (2.22)]{Hi}, the Majorana field can be expressed in terms of creation and annihilation operators, as a quadratic form, as
\[
\psi(x)= \frac1{(2\pi)^{3/2}} \int_{\bR^3} d\bdp\bigg( e^{ipx} \sum_{s = \pm \frac12} u_s(\bdp)a_s(\bdp)^* + e^{-ipx} \sum_{s = \pm \frac12} v_s(\bdp)a_s(\bdp)\bigg), 
\]
with 4-component spinors $u_s(\bdp), v_s(\bdp)$, $s = \pm \frac12$, defined by
\[
u_{-\frac12}(\bdp) = \left[ \begin{matrix}\nu_0(\bdp) \\ 0\end{matrix}\right], \quad u_{\frac12}(\bdp) = \left[ \begin{matrix}0 \\ -\sigma_2 \overline{\nu_0(\bdp)} \end{matrix}\right], \quad  v_{-\frac12}(\bdp) = \left[ \begin{matrix}0 \\ \sigma_2 \overline{\nu_0(\bdp)} \end{matrix}\right], \quad v_{\frac12}(\bdp) = \left[ \begin{matrix}-\nu_0(\bdp) \\ 0\end{matrix}\right],
\]
and where the creation operators $a_s^*$, $s = \pm\frac12$, are such that $a(l,h)^* = a_{\frac12}(l)^* + a_{-\frac12}(h)^*$, $l, h \in L^2(\bR^3)$. It is then straightforward to compute
\begin{equation}\label{eq:psidpsi}
\langle \Psi, : \psi^\dagger(0,\bdx) \gamma^0\gamma^k \psi(0,\bdy):\Psi\rangle = \frac{2}{(2\pi)^3} 
\big[ v(\bdx)^\dagger \gamma^0 \gamma^k v(\bdy) - u(\bdx)^\dagger \gamma^0 \gamma^k u(\bdy)\big]
\end{equation}
with
\[
u(\bdx) := \int_{\bR^3} d\bdp\, e^{-i \bdp\cdot \bdx} \sum_{s=\pm1/2} \overline{l_s(\bdp)} u_s(\bdp), \quad 
v(\bdx) := \int_{\bR^3} d\bdp\, e^{i \bdp\cdot \bdx} \sum_{s=\pm1/2} l_s(\bdp) v_s(\bdp).
\]
Moreover, by \eqref{eq:sigma2p} and \eqref{eq:nu0dyad}, one obtains
\[\begin{split}
\sum_{s=\pm \frac12} \overline{l_s(\bdp)} u_s(\bdp) &= \frac1{(2\pi)^{3/2}} \left[ \begin{matrix} \nu_0(\bdp)\nu_0(\bdp)^\dagger \hat \Phi_0(-\bdp) \\ \sigma_2 \overline{\nu_0(\bdp)} \nu_0(\bdp)^t \overline{\hat\Phi_0(\bdp)} \end{matrix}\right]
= \frac1{2(2\pi)^{3/2} |\bdp|}\left[\begin{matrix} \undertilde{p}_+ & 0 \\ 0 & \tilde{p}_+\end{matrix}\right] \hat\Psi_0(-\bdp)\\
&=  \frac1{2(2\pi)^{3/2}}\left(1+\frac{p_j}{|\bdp|}\gamma^0 \gamma^j\right) \hat\Psi_0(-\bdp),
\end{split}\] 
and similarly
\[
\sum_{s=\pm \frac12} l_s(\bdp) v_s(\bdp) = - \frac1{2(2\pi)^{3/2}}\left(1+\frac{p_j}{|\bdp|}\gamma^0 \gamma^j\right) \hat\Psi_0(\bdp).
\]
Inserting these results into Eq.~\eqref{eq:psidpsi} and using the identity $\gamma^0 \gamma^j \gamma^0 \gamma^k =\delta_{jk} + i \sigma^{jk} = \delta_{jk} +i \varepsilon_{jkh}\Sigma_h$, yields
\[\begin{split}
\langle \Psi, &: \psi^\dagger(0,\bdx) \gamma^0\gamma^k \psi(0,\bdy):\Psi\rangle =\\
&=\frac{1}{(2\pi)^6} \int_{\bR^6} d\bdp\,d\bdq\,e^{-i(\bdp\cdot\bdx-\bdq\cdot\bdy)} \hat\Psi_0(\bdp)^\dagger\left(\frac{p_j}{|\bdp|} \gamma^0\gamma^j\gamma^0\gamma^k + \frac{q_j}{|\bdq|} \gamma^0\gamma^k\gamma^0\gamma^j\right)\hat\Psi_0(\bdq)\\
&=\frac{1}{(2\pi)^6} \int_{\bR^6} d\bdp\,d\bdq\, e^{-i(\bdp\cdot\bdx-\bdq\cdot\bdy)} \hat\Psi_0(\bdp)^\dagger\left[\frac{p_k}{|\bdp|}+\frac{q_k}{|\bdq|}+i \varepsilon_{jkh} \left( \frac{p_j}{|\bdp|} -\frac{q_j}{|\bdq|}\right) \Sigma_h \right]\hat\Psi_0(\bdq)
\end{split}\]
The formula in the statement is then easily obtained from
\[
\langle \Psi ,T_{00}(0,\bdx)\Psi\rangle = \frac i 2 \left(\frac{\partial}{\partial x_k}-\frac{\partial}{\partial y_k}\right)\langle \Psi, :\psi^\dagger(0,\bdx) \gamma^0\gamma^k \psi(0,\bdy):\Psi\rangle |_{\bdy=\bdx}.
\]
The final part of the statement is a consequence of the fact that, from this formula, $\langle \Psi, T_{00}(0,\bdx) \Psi\rangle$ can be expressed as a sum of terms of the form
\[
(G * P_1 \Psi_0)(\bdx)^\dagger P_2\Psi_0(\bdx) \quad \text{or} \quad P'_1\Psi_0(\bdx)^\dagger (G'*P'_2 \Psi_0)(\bdx),
\] 
with $G, G'$ tempered distributions and $P_j, P_j'$, $j=1,2$, differential operators.
\end{proof}



\begin{thebibliography}{99}

\bibitem{Ar}H.~Araki, Relative entropy of states of von Neumann algebras, \emph{Publ. RIMS} \textbf{11} (1976), 809-833.
%
\bibitem{ABCH} R.~E.~Arias, D.~D.~Blanco, H.~Casini, M.~Huerta, Local tempertures and local terms in modular Hamiltonians, \emph{Phys. Rev. D} \textbf{95} (2017), 065005.
%
\bibitem{BW1}J.~Bisognano, E.~Wichmann, On the duality condition for a Hermitian scalar field, \emph{J.\ Math.\ Phys.} \textbf{16} (1975), 985-1007.
%
\bibitem{BW2}J.~Bisognano, E.~Wichmann, On the duality condition for quantum fields, \emph{J.\ Math.\ Phys.} \textbf{17} (1976), 303-321. 
%
\bibitem{BGL}R.~Brunetti, D.~Guido, R.~Longo, Modular structure and duality in conformal quantum field theory, \emph{Commun. Math. Phys.} \textbf{156} (1993), 201-219.
%
\bibitem{Bu}D.~Buchholz, On the structure of local quantum fields with non-trivial interaction, in Proc.\ Intern.\ Conf.\ Operator Algebras, Ideals and Their Applications in Physics, H.~Baumg\"artel ed. (Teubner, 1977), 146-153
%
\bibitem{CGP}H.~Casini, S.~Grillo, D.~Pontello, Relative entropy for coherent states from Araki formula, \emph{Phys. Rev. D} \textbf{99} (2019), 125020.
%
\bibitem{CH}H.~Casini, M.~Huerta,  Entanglement entropy in free QFT, \emph{J. Phys. A} \textbf{42} (2009), 504007.
%
\bibitem{CF}F.~Ceyhan, T.~Faulkner, Recovering the QNEC from the ANEC, \emph{Commun. Math. Phys.} \textbf{377} (2020), 999-1045.
%
\bibitem{CLR}F.~Ciolli, R.~Longo, G.~Ruzzi, The information in a wave, \emph{Commun. Math. Phys.} \textbf{379} (2020), 979-1000.
%
\bibitem{CM}R.~Conti, G.~Morsella, Quasi-free isomorphisms of second quantization algebras and modular theory, arXiv:2305.07606, to apper on \emph{Math. Phys. Anal. Geom.}
%
\bibitem{Di}J.~Dimock, Algebras of local observables on a manifold, \emph{Commun. Math. Phys.} \textbf{77} (1980), 219-228.
%
\bibitem{GMV}S.~Galanda, A.~Much, R.~Verch, Relative entropy of fermion excitation states on the CAR algebra, \emph{Math. Phys. Anal. Geom.} \textbf{26} (2023), 21-38 (?)
%
\bibitem{Gi}M.~Gillioz, Conformal Field Theory for Particle Physicists: From QFT Axioms to the Modern Conformal Bootstrap, SpringerBriefs in Physics, 2023. 
%
\bibitem{Ha}R.~Haag, \emph{Local quantum physics. Fields, particles, algebras}, second edition, Springer-Verlag, Berlin, 1996.
\bibitem{HL}P.~D.~Hislop, R.~Longo, Modular structure of the local algebras associated with the
free massless scalar field theory, \emph{Commun. Math. Phys.} \textbf{84} (1982), 71-85.
\bibitem{Hi}P.~D.~Hislop, Conformal covariance, modular structure, and duality for local algebras in free massless quantum field theories, \emph{Ann. Phys.} \textbf{185} (1988), 193-230.
%
\bibitem{Lo}R.~Longo, Entropy of coherent excitations, \emph{Lett.\ Math.\ Phys.}\ \textbf{109} (2019), 2587-2600.
%
\bibitem{LM} R.~Longo, G.~Morsella, The massless modular hamiltonian,  \emph{Commun. Math. Phys.} \textbf{400} (2023), 1181-1201.
%
\bibitem{OP}M.~Ohya, D.~Petz, Quantum entropy and its use, Springer, 1993.
\bibitem{Ta} M.~Takesaki, Theory of operator algebras, I-III, Springer-Verlag, New-York-Heidelberg, 1979, 2003.
\bibitem{Tu}F.~Tuteri, Fenomenologia modulare, Bachelor thesis in Physics, University of Roma Tor Vergata, 2022.
%
\bibitem{Wi}E.~Witten, APS medal for exceptional achievement in research: invited article on entanglement properties of quantum field theory, \emph{Rev.\ Mod.\ Phys.}\ \textbf{90} (2018), 045003.
\end{thebibliography}
\end{document}